\documentclass[10pt,twocolumn,twoside,journal]{IEEEtran}

\usepackage[noadjust]{cite}

\ifCLASSINFOpdf
  \usepackage[pdftex]{graphicx}
  \DeclareGraphicsExtensions{.pdf,.jpeg,.png,.tif,.tiff}
\else

\fi

\usepackage[cmex10]{amsmath}

\ifCLASSOPTIONcompsoc
   \usepackage[tight,normalsize,sf,SF]{subfigure}
\else
   \usepackage[tight,footnotesize]{subfigure}
\fi

\usepackage[english]{babel}
\usepackage{amsthm, amsfonts, amssymb, mathrsfs, amsmath}

\usepackage{color}
\usepackage{graphicx}
\usepackage{nicefrac}
\usepackage{mathtools}
\usepackage{multirow}
\usepackage[ruled]{algorithm2e}
\usepackage{calrsfs}
\usepackage[mathcal]{euscript}

\usepackage[draft]{fixme}
\usepackage{enumerate}
\usepackage{dsfont}

\newcommand*{\mathcolor}{}
\def\mathcolor#1#{\mathcoloraux{#1}}
\newcommand*{\mathcoloraux}[3]{%
  \protect\leavevmode
  \begingroup
    \color#1{#2}#3%
  \endgroup
}

\usepackage{tikz}
\allowdisplaybreaks[1]

\newcommand{\bx}{{\boldsymbol{x}}}
\newcommand{\bu}{{\boldsymbol{u}}}
\newcommand{\bv}{{\boldsymbol{v}}}
\newcommand{\bz}{{\boldsymbol{z}}}
\newcommand{\by}{\boldsymbol{y}}

\newcommand{\bhx}{{\boldsymbol{\widehat{x}}}}

\newcommand{\g}{\textsl{g}}

\newcommand{\bg}{\textbf{\textsl{g}}}

\makeatletter
\renewcommand{\@algocf@capt@plain}{above}
\makeatother

\DeclareMathOperator*{\argmin}{arg\,min}


\theoremstyle{definition}

\usepackage{amsthm}
\newtheorem{theorem}{Theorem}[section]
\newtheorem{corollary}{Corollary}[theorem]
\newtheorem{lemma}[theorem]{Lemma}

\newtheorem{proposition}[theorem]{Proposition}
\newtheorem{assumption}[theorem]{Assumption}

\begin{document}

\title{\textcolor{black}{Measurement} Bounds for Sparse Signal \textcolor{black}{Recovery} With Multiple Side Information}

\author{Huynh Van Luong,~J\"{u}rgen Seiler,~Andr\'{e} Kaup,~S{\o}ren Forchhammer,~and~Nikos Deligiannis

\thanks{H. V. Luong, J. Seiler, and A. Kaup are with the Chair of Multimedia Communications and Signal Processing, Friedrich-Alexander-Universit\"{a}t Erlangen-N\"{u}rnberg, 91058 Erlangen, Germany (e-mail: huynh.luong@fau.de, juergen.seiler@fau.de, and andre.kaup@fau.de).}
\thanks{S. Forchhammer is with the Department of Photonics Engineering, Technical University of Denmark, 2800 Lyngby, Denmark (e-mail: sofo@fotonik.dtu.dk).}
\thanks{N. Deligiannis is with the Department of Electronics and Informatics, Vrije Universiteit Brussel, 1050 Brussels, and also with iMinds, 9050 Ghent, Belgium (e-mail: ndeligia@etro.vub.ac.be).}
}

%
\maketitle

\begin{abstract}
  In the context of compressed sensing (CS), this paper considers the problem of reconstructing sparse signals with the aid of other given correlated sources as multiple side information. To address this problem, we theoretically study a generic \textcolor{black}{weighted $n$-$\ell_{1}$ minimization} framework and propose a reconstruction algorithm that leverages multiple side information signals (RAMSI). The proposed RAMSI algorithm computes adaptively optimal weights among the side information signals at every reconstruction iteration. In addition, we establish theoretical bounds on the number of measurements that are required to successfully reconstruct the sparse source by using \textcolor{black}{weighted $n$-$\ell_{1}$ minimization}. The analysis of the established bounds reveal that \textcolor{black}{weighted $n$-$\ell_{1}$ minimization} can achieve sharper bounds and significant performance improvements compared to classical CS. We evaluate experimentally the proposed RAMSI algorithm and the established bounds using synthetic sparse signals as well as correlated feature histograms, extracted from a multiview image database for object recognition. The obtained results show clearly that the proposed algorithm outperforms state-of-the-art algorithms---\textcolor{black}{including classical CS, $\ell_1\text{-}\ell_1$ minimization, Modified-CS, regularized Modified-CS, and weighted $\ell_1$ minimization}---in terms of both the theoretical bounds and the practical performance.
\end{abstract}

\vspace{5pt}

\ifCLASSOPTIONpeerreview
\begin{center} \bfseries EDICS Category: SAM-CSSM \end{center}
\fi

\IEEEpeerreviewmaketitle

\section{Introduction}\label{sec:intro}
\IEEEPARstart{C}{ompressed} 
sensing (CS) is a theory to perform sparse signal reconstruction that has attracted significant attention \cite{DonohoTIT06,candes2006robust,DonohoCOM06,CandesTIT06,Beck09,Candes08,Asif13,NVaswani09,NVaswani10,NVaswaniTSP10,VaswaniZ16,MPFriedlander12,JZhan14,JZhan15,AminTIT15,Weizman15,Venkat12,Tropp14,BaronARXIV09,DuarteTIT13} in the past decade. CS enables sparse signals to be recovered in a computationally tractable manner from a relatively limited number of random measurements. This can be done by solving a Basis Pursuit problem, which involves the $\ell_{1}$-norm minimization of the sparse signal subject to the measurements. It has been shown that sparse signal reconstruction can be further improved by replacing the $\ell_{1}$-norm with a weighted $\ell_{1}$-norm \cite{Candes08,Asif13,AminTIT15,Weizman15}. The works in \cite{Venkat12,Tropp14} give general conditions for exact and robust recovery \textcolor{black}{with the purpose to provide accurate CS bounds on the number of measurements} required for successful reconstruction based on convex optimization problems. Furthermore, distributed compressive sensing \cite{BaronARXIV09,DuarteTIT13} allows a correlated ensemble of sparse signals to be jointly recovered by exploiting intra- and inter-signal correlations.

\textcolor{black}{
We focus on the problem of efficiently reconstructing a sparse source from low-dimensional random measurements given side (alias, prior) information. In a practical setting, when reconstructing a signal, we may have access to side or prior information, gleaned from knowledge on the structure of the signal or from a known correlated signal that has spatial or temporal similarities with the target signal. The problem of signal recovery with side (or prior) information was initially studied in~\cite{NVaswani09,NVaswaniTSP10,NVaswaniTSP12,ZhangR11,KhajehnejadXAH09,KhajehnejadXAH11}. Specifically, the modified-CS method in~\cite{NVaswani09,NVaswaniTSP10} considered that a part of the support is available from prior knowledge and tried to find the signal that satisfies the measurement constraint and is sparsest outside the known support. Prior information
on the sparsity pattern of the data was also considered in~\cite{ZhangR11} and information-theoretic guarantees were presented. The studies in \cite{KhajehnejadXAH09,KhajehnejadXAH11} introduced weights into the $\ell_{1}$ minimization framework that depend on the partitioning of the source signal into two sets, with the entries of each set having a specific probability of being nonzero.}

Alternatively, the studies in~\cite{MotaGLOBALSIP14,MotaARXIV14} incorporated side information on CS by means of~$\ell_{1}$\text{-}$\ell_{1}$ minimization and derived bounds on the number of Gaussian measurements required to guarantee perfect
signal recovery. It was shown---both by theory and experiments---that $\ell_1\text{-}\ell_1$ minimization can dramatically
improve the reconstruction performance over CS subject to a good-quality side information~\cite{MotaGLOBALSIP14,MotaARXIV14}.
\textcolor{black}{
The study in~\cite{MansourS15} proposed a weighted-$\ell_{1}$
minimization method to incorporate prior information---in the form of inaccurate support estimates---into CS. The work also provided
bounds on the number of Gaussian measurements for successful recovery when
sufficient support information is available.}
Furthermore, recent studies investigated the performance of CS with single side information in practical applications, namely, compressive video foreground extraction~\cite{MotaICASSP15,MotaARXIV15}, magnetic resonance imaging~\cite{Weizman15}, and mm-wave synthetic aperture radar imaging~\cite{Becquaert2016Radar}.
For instance, the work in \cite{Weizman15} exploits the temporal similarity
in MRI longitudinal scans to accelerate MRI acquisition. In the application
scenario of video background subtraction \cite{MotaICASSP15,MotaARXIV15,Scarlett},
prior information, which is generated from a past frame, is used to reduce
the number of measurements for a subsequent frame.

\textcolor{black}{The problem of sparse signal recovery with prior information also emerges in the context of reconstructing a time sequence of sparse signals from low-dimensional measurements~\cite{ZhangR11,AsifR14,QiuV11,GuoQV14,QiuVLH14,VaswaniZ16}. The study in~\cite{VaswaniZ16} provided a comprehensive overview of the domain reviewing a class of recursive algorithms for recovering a time sequence of sparse signals from a small number of measurements. The problem of estimating a time-varying, sparse signal from streaming
measurements was studied in~\cite{VaswaniZ16,AsifR14}, while the work in \cite{ZhangR11} addressed the recovery problem in the context
of multiple measurement vectors. The problem also appears in the context of robust PCA and online robust PCA~\cite{QiuV11,GuoQV14,QiuVLH14}, a framework that finds important application in video background subtraction. The studies in~\cite{QiuV11,GuoQV14,QiuVLH14} used the modified-CS~\cite{NVaswaniTSP10,NVaswaniTSP12} method to leverage prior knowledge under the condition of slow support and signal value changes.}
\vspace{-7pt}
\textcolor{black}{
\subsection{Motivation}
Recent emerging applications \cite{Taj2011,ChellappaIEEE08,YangIEEE10}, such as visual sensor surveillance and mobile augmented reality, are following a distributed sensing scenario where a plethora of tiny heterogeneous devices collect information from the environment. In certain scenarios, we may deal with very high dimensional data where sensing and processing are reasonably expensive under time and resource constraints. These challenges can be addressed by leveraging the distributed sparse representation of the multiple sources \cite{YangIEEE10}. The problem in this setup is to represent and reconstruct the sparse sources along with exploiting the correlation among them. One of the key questions is how to robustly reconstruct a compressed source from a small number of measurements given side information gleaned from multiple other available sources.}

Existing attempts to incorporate side information in compressed sensing are typically considering one side information signal that is typically of good quality; see for example~\cite{MotaGLOBALSIP14,MotaARXIV14,MotaICASSP15,MotaARXIV15,Weizman15}. However, we are aiming at robustly reconstructing the compressed source given multiple side information signals in scenarios where the multiple sources are changing in time, that is, there are arbitrary side information qualities and unknown correlations among them. The challenge raises key interesting questions for solving the distributed sensing problem:
\begin{itemize}
  \item How can we take advantage of the multiple side information signals? This implies an optimal strategy able to effectively exploit the useful information from the multiple signals as well as adaptively eliminate negative effects of poor side information samples.
  \item How many measurements are required to successfully reconstruct the sparse source given multiple side information signals? This calls for bounds on the number of measurements required to guarantee successful signal recovery.
\end{itemize}
\vspace{-7pt}
\subsection{Contributions}
To address the aforementioned challenges, we propose a reconstruction strategy that leverages multiple side information signals leading to higher signal recovery performance compared to state-of-the-art methods~\cite{Beck09,MotaGLOBALSIP14,NVaswaniTSP10,NVaswaniTSP12,KhajehnejadXAH11}. This paper contributes in a twofold way. \textcolor{black}{Firstly, an efficient reconstruction algorithm that leverages multiple side information (RAMSI) is proposed\nocite{LuongDCC16}. The second contribution is to theoretically establish measurement bounds serve as bounds for the number of measurements required by the
proposed RAMSI algorithm.}

\textcolor{black}{The proposed algorithm solves a generic re-weighted $n$-$\ell_{1}$ minimization problem:}
Per iteration of the reconstruction process, the algorithm adaptively selects optimal weights for the multiple signals. As such, contrary to existing works \cite{MotaGLOBALSIP14,MotaARXIV14,MotaICASSP15,MotaARXIV15,Weizman15}, which exploit only one side information signal, the proposed algorithm can efficiently leverage the correlations among multiple heterogenous sources and adapt on-the-fly to changes of the correlations.

We also establish the measurement bounds for the \textcolor{black}{weighted $n\text{-}\ell_1$ minimization problem, which serve as bounds for the proposed RAMSI algorithm.} The bounds depend on the support of the source signal to be recovered and the correlations between the source signal and the multiple side information signals. The correlations are expressed via the supports of the differences between the source and side information signals. We will show that the weighted $n$-$\ell_{1}$ minimization bounds are sharper compared to those of the classical CS \cite{DonohoTIT06,DonohoCOM06,CandesTIT06} and the $\ell_{1}$-$\ell_{1}$ reconstruction \cite{MotaGLOBALSIP14,MotaARXIV14} methods. These theoretical bounds evidently depict the advantage
of RAMSI to deal with heterogeneous side information signals including possible poor side information signals. Furthermore, we show---both theoretically and practically---that the performance of the method is improved with the number of available side information signals.
\vspace{-7pt}
\subsection{Outline}
The rest of this paper is organized as follows. Section \ref{relatedWork} reviews the background on the related theory as well as the corresponding \textcolor{black}{Gaussian measurement} bounds for CS and CS with side information. The RAMSI algorithm is proposed in Section \ref{RAMSI}, while our bounds and the corresponding analysis are presented in Section \ref{ramsiBounds}. The derived bounds and the performance of RAMSI on different sparse sources are assessed in Section \ref{Experiment} and Section~\ref{conclusion} concludes the work.

\section{\textcolor{black}{Background}}\label{relatedWork}
In this section we review the problem of signal recovery from low-dimensional measurements \cite{DonohoCOM06,CandesTIT06,Beck09,Candes08,Asif13} represented by CS (Sec. \ref{fundamentalRecovery}) and CS with prior information \cite{MotaGLOBALSIP14,MotaARXIV14,Scarlett,MotaICASSP15,Weizman15,WarnellTIP15} (Sec. \ref{weightL1-CSwithSI}). In addition, we present background information with respect to measurement bounds (Sec. \ref{backgroundConvex}).

\subsection{Sparse Signal Recovery}\label{sparseRecovery}
\subsubsection{Compressed Sensing}
\label{fundamentalRecovery}
The problem of low-dimensional signal recovery arises in a wide range of applications such as statistical inference and signal processing. Most signals in such applications have sparse representations in some domain or learned set of basis. Let $\bx\in\mathbb{R}^{n}$ denote a high-dimensional sparse vector. The source $\bx$ can be reduced by sampling via a linear projection at the encoder~\cite{CandesTIT06}. We denote a random measurement matrix by $\mathbf{\Phi}\in \mathbb{R}^{ m\times n}$ (with~$m\ll{n}$), whose elements are sampled from an i.i.d. Gaussian distribution. Thus, we get a measurement vector $\by=\mathbf{\Phi}\bx$, with $\by\in\mathbb{R}^{ m}$. At the decoder, $\bx$ can be recovered by solving the Basis Pursuit problem~\cite{CandesTIT06,DonohoCOM06}:
\begin{equation}\label{l1-norm}
    \min_{\bx} \|\bx\|_{1} \mathrm{~subject~to~} \by=\mathbf{\Phi}\bx,
\end{equation}
where $\|\bx\|_{p}:=(\sum_{i=1}^{n}|x_{i}|^{p})^{1/p}$ is the $\ell_{p}$-norm of $\bx$ wherein $x_{i}$ is an element of $\bx$.
Problem \eqref{l1-norm} becomes an instance of finding a general solution:
\begin{equation}\label{l1-general}
    \min_{\bx}\{H(\bx) = f(\bx) + g(\bx)\},
\end{equation}
where $f:=\mathbb{R}^{n} \rightarrow\mathbb{R}$ is a smooth convex function and~$g:=\mathbb{R}^{n} \rightarrow\mathbb{R}$ is a continuous convex function, which is possibly non-smooth. Problem \eqref{l1-norm} emerges from Problem~\eqref{l1-general} when setting $g(\bx)=\lambda \|\bx\|_{1}$ and $f(\bx)=\frac{1}{2}\|\mathbf{\Phi}\bx-\by\|^{2}_{2}$, with \textcolor{black}{Lipschitz constant} $L_{\nabla f}$ \cite{Beck09}. 
Using proximal gradient methods~\cite{Beck09}, $\bx^{(k)}$ at iteration $k$ is computed as
\begin{equation}\label{l1-proximal}
    \bx^{(k)}= \Gamma_{\frac{1}{L}g}\Big(\bx^{(k-1)}-\frac{1}{L}\nabla f(\bx^{(k-1)})\Big),
\end{equation}
where $L\geq L_{\nabla f}$ and $\Gamma_{\frac{1}{L}g}(\bx)$ is a proximal operator defined as
\begin{equation}\label{l1-proximalOperator}
    \Gamma_{\frac{1}{L}g}(\bx) = \argmin_{\bv \in\mathbb{R}^{n}}\Big\{ \frac{1}{L}g(\bv) + \frac{1}{2}\|\bv-\bx\|^{2}_{2}\Big\}.
\end{equation}

The classical $\ell_{1}$ minimization problem in CS \cite{DonohoTIT06,DonohoCOM06,CandesTIT06} requires $m_{\ell_{1}}$ measurements \cite{Venkat12,MotaGLOBALSIP14,MotaARXIV14} for successful reconstruction, bounded as
 \begin{equation}\label{l1 bound}
    m_{\ell_{1}} \geq 2s_{0}\log\frac{n}{s_{0}} + \frac{7}{5}s_{0} + 1,
\end{equation}
where $s_{0}:=\|\bx\|_0=|\{i: x_{i}\neq 0\}|$ denotes the number of nonzero elements in $\bx$ as the support of $\bx$, with $|.|$ denoting the cardinality of a set and $\|\cdot\|_0$ being the $\ell_0$-pseudo-norm.
\subsubsection{CS with $\ell_{1}$-$\ell_{1}$ Minimization}
\label{weightL1-CSwithSI}

The $\ell_{1}$-$\ell_{1}$ minimization approach~\cite{MotaGLOBALSIP14,MotaARXIV14,MotaICASSP15} reconstructs $\bx$ given a signal $\bz \in \mathbb{R}^{n}$ as side information by solving the following problem:
\begin{equation}\label{l1-l1minimization}
    \min_{\bx}\Big\{\frac{1}{2}\|\mathbf{\Phi}\bx-\by\|^{2}_{2} + \lambda (\|\bx\|_{1}+\|\bx-\bz\|_{1})\Big\}.
\end{equation}

The bound on the number of measurements required by Problem~\eqref{l1-l1minimization} to successfully reconstruct $\bx$ depends on the quality of the side information signal $\bz$ as \cite{MotaGLOBALSIP14,MotaARXIV14,MotaICASSP15}
\begin{equation}\label{l1-l1 bound}
    m_{\ell_{1}\text{-}\ell_{1}} \geq 2\overline{h}\log\Big(\frac{n}{s_{0}+\xi/2}\Big) + \frac{7}{5}\Big(s_{0}+\frac{\xi}{2}\Big) + 1,
\end{equation}
where
\begin{subequations}\label{l1-l1 sparse set}
\begin{align}
\xi:&=|\{i:z_{i}\neq x_{i}=0\}|-|\{i:z_{i}= x_{i}\neq 0\}|\label{l1-l1 sparse setXi} \\
\bar{h}:&=\mathcolor{black}{|\{i:x_{i}>0,x_{i}>z_{i}\}\cup\{i:x_{i}<0,x_{i}<z_{i}\}|}\label{l1-l1 sparse setHBar},
\end{align}
\end{subequations}
wherein $x_{i}$, $z_{i}$ are corresponding elements of $\bx$, $\bz$. It has been shown that Problem \eqref{l1-l1minimization} improves over Problem \eqref{l1-norm} provided that the side information has good enough quality~\cite{MotaGLOBALSIP14,MotaARXIV14}. The quality is expressed by a high number of elements $z_{i}$ that are equal to $x_{i}$, thereby leading to $\xi$ in \eqref{l1-l1 sparse setXi} being small.

\subsection{Background on Measurement Bounds}\label{backgroundConvex}
We introduce some key definitions and conditions in convex optimization as well as linear inverse problems, based on the concepts in \cite{Venkat12,Tropp14}, which are used in the derivation of the measurement bounds for the proposed \textcolor{black}{weighted $n\text{-}\ell_{1}$ minimization approach}.
\subsubsection{Convex Cone}
A \textit{convex cone} $C\subset \mathbb{R}^{n}$ is a convex set that satisfies $C=\tau C$, $\forall\tau\geq0$~\cite{Tropp14}. For the cone $C\subset \mathbb{R}^{n}$, a \textit{polar cone} $C^{\circ}$ is the set of outward normals of $C$, defined by
\begin{equation}\label{polarCone}
C^{\circ}:=\{\bu\in \mathbb{R}^{n}:\bu^{\mathrm{T}}\bx\leq 0, ~\forall\bx\in C\}.
\end{equation}
A \textit{descent cone}~\cite[Definition 2.7]{Tropp14} $\mathcal{D}(g,\bx)$, alias \textit{tangent cone}~\cite{Venkat12}, of a convex function $g:=\mathbb{R}^{n} \rightarrow\mathbb{R}$ at a point $\bx\in \mathbb{R}^{n}$---at which $g$ is not increasing---is defined as
\begin{equation}\label{descentCone}
\mathcal{D}(g,\bx):=\bigcup\limits_{\tau \geq 0}\{\by\in \mathbb{R}^{n}:g(\bx+\tau \by)\leq g(\bx)\},
\end{equation}
where $\bigcup$ denotes the union operator.

\subsubsection{Gaussian Width}
The Gaussian width \cite{Venkat12} is a summary parameter for convex cones; it is used to measure the aperture of a convex cone. For a convex cone $C\subset \mathbb{R}^{n}$, considering a subset $C\cap \mathbb{S}^{n-1}$ where $\mathbb{S}^{n-1}\subset \mathbb{R}^{n}$ is a unit sphere, the \textit{Gaussian width} \cite[Definition 3.1]{Venkat12} is defined as
\begin{equation}\label{gaussianWidth}
\omega(C):=\mathbb{E}_{\bg}[\sup_{\bu \in C\cap \mathbb{S}^{n-1}} \textbf{\textsl{g}}^{\mathrm{T}}\bu].
\end{equation}
\textcolor{black}{where $\bg\sim\mathcal{N}(0,\mathbf I_{n})$ is a vector of $n$ independent, zero-mean, and unit-variance Gaussian random variables and $\mathbb{E}_{\bg}[\cdot]$ denotes the expectation with respect to $\bg$.}
The Gaussian width~\cite[Proposition 3.6]{Venkat12} can further be bounded as
\begin{equation}\label{gaussianWidthBound}
  \omega(C)\leq \mathbb{E}_{\bg}\big[\mathrm{dist}(\textbf{\textsl{g}},C^{\circ})],
\end{equation}
where $\mathrm{dist}(\textbf{\textsl{g}},C^{\circ})$ denotes the Euclidean distance of $\textbf{\textsl{g}}$ with respect to the set $C^{\circ}$, which is in turn defined as
\begin{equation}\label{euclideanDistance}
  \mathrm{dist}(\textbf{\textsl{g}},C^{\circ}):=\min_{\bu}\{\|\textbf{\textsl{g}}-\bu\|_{2}:\bu\in C^{\circ}\}.
\end{equation}

Recently, a new summary parameter called the \textit{statistical dimension} $\delta(C)$ of cone $C$ \cite{Tropp14} is introduced to estimate the convex cone~\cite[Theorem
4.3]{Tropp14}. Using the Gaussian width, the statistical dimension is bounded as~\cite[Proposition 10.2]{Tropp14}
\begin{equation}\label{gaussianWidthStatDim}
  \omega^{2}(C)\leq \delta(C) \leq \omega^{2}(C)+1.
\end{equation}
This relationship gives a convenient bound for the Gaussian width that is to be used in our following computations. The statistical dimension can be expressed in terms of the polar cone $C^{\circ}$ as~\cite[Proposition 3.1]{Tropp14}
\begin{equation}\label{statDimComputePolar}
  \delta(C):=\mathbb{E}_{\bg}\big[\mathrm{dist}^{2}(\textbf{\textsl{g}},C^{\circ})].
\end{equation}
\subsubsection{Measurement Condition}
An optimality condition~\cite[Proposition 2.1]{Venkat12}, \cite[Fact 2.8]{Tropp14} for linear inverse problems states that $\bx_{0}$ is the unique solution of \eqref{l1-general} if and only if
\begin{equation}\label{optimalCondition}
  \mathcal{D}(g,\bx_{0})\cap \mathrm{null}(\mathbf{\Phi})=\{\mathbf{0}\},
\end{equation}
where $\mathrm{null}(\mathbf{\Phi})\hspace{-2pt}:=\hspace{-2pt}\{\bx\in\mathbb{R}^{n}\hspace{-2pt}:\hspace{-2pt}\mathbf{\Phi}\bx\hspace{-2pt}=\hspace{-2pt}\mathbf{0}\}$ is the null space of $\mathbf{\Phi}$. We consider the number of measurements $m$ required to successfully reconstruct a given signal $\bx_{0}\in \mathbb{R}^{n}$. Corollary 3.3 in \cite{Venkat12} states that, given a measurement vector $\by=\mathbf{\Phi}\bx_{0}$, $\bx_{0}$ is the unique solution of \eqref{l1-general} with probability at least $1\hspace{-2pt}-\hspace{-2pt}\exp(-\frac{1}{2}(\sqrt{m}\hspace{-2pt}-\hspace{-2pt}\omega(\mathcal{D}(g,\bx_{0})))^{2})$ provided that $m\geq \omega^{2}(\mathcal{D}(g,\bx_{0}))+1$. Furthermore, combined with the relationship in \eqref{gaussianWidthStatDim}, we can interpret the successful recovery of $\bx_{0}$ in an equivalent condition:
 \begin{equation}\label{measurementCondition}
  m\geq \delta(\mathcal{D}(g,\bx_{0}))+1.
\end{equation}
\subsubsection{Bound on the Measurement Condition}
The key remaining question is how to calculate the statistical dimension $\delta(\mathcal{D}(g,\bx))$ of a descent cone $\mathcal{D}(g,\bx)$.
Using \eqref{statDimComputePolar}, we can calculate $\delta(\mathcal{D}(g,\bx))$ as
 \begin{equation}\label{dimDescent}
     \delta(\mathcal{D}(g,\bx))=\mathbb{E}_{\bg}\big[\mathrm{dist}^{2}(\textbf{\textsl{g}},\mathcal{D}(g,\bx)^{\circ})\big],
\end{equation}
 where $\mathcal{D}(g,\bx)^{\circ}$ is the polar cone of $\mathcal{D}(g,\bx)$ as defined in \eqref{polarCone}. Let us consider that the subdifferential $\partial g$ \cite{JHiriat} of a convex function $g$ at a point $\bx \in \mathbb{R}^{n}$ is given by $\partial g\hspace{-2pt}:=\hspace{-2pt}\{\bu \hspace{-2pt}\in\hspace{-2pt}\mathbb{R}^{n}\hspace{-2pt}: g(\by)\hspace{-2pt}\geq\hspace{-2pt} g(\bx)\hspace{-2pt}+\hspace{-2pt}\bu^{\mathrm{T}}(\by\hspace{-2pt}-\hspace{-2pt}\bx)~\text{for all}~ \by \hspace{-2pt}\in \hspace{-2pt}\mathbb{R}^{n}$\}. From \eqref{dimDescent} and~\cite[Proposition 4.1]{Tropp14}, we obtain an upper bound on $\delta(\mathcal{D}(g,\bx))$ as
 \begin{equation}\label{upperBound}
 \begin{split}
     \delta(\mathcal{D}(g,\bx))&=\mathbb{E}_{\bg}\big[\min_{\tau\geq 0}\mathrm{dist}^{2}(\textbf{\textsl{g}},\tau \hspace{-2pt}\cdot \hspace{-2pt} \partial g(\bx))\big] \\
     &\leq \min_{\tau\geq 0}\mathbb{E}_{\bg}\big[\mathrm{dist}^{2}(\textbf{\textsl{g}},\tau \hspace{-2pt}\cdot \hspace{-2pt} \partial g(\bx))\big].
\end{split}
\end{equation}
In short, we conclude the following proposition.
\begin{proposition}[Measurement bound for a convex norm function]\label{propUpper}
In order to obtain the measurement bound for the recovery condition, $m\geq U_{g}\hspace{-1pt}+\hspace{-1pt}1$, we calculate the quantity $U_{g}$ given a convex norm function $g\hspace{-1pt}:=\hspace{-1pt}\mathbb{R}^{n} \hspace{-2pt}\rightarrow\hspace{-2pt}\mathbb{R}$ by\begin{equation}\label{upperBoundCompute}
U_{g}= \min_{\tau\geq 0}\mathbb{E}_{\bg}\Big[\mathrm{dist}^{2}(\textbf{\textsl{g}}, \tau \cdot\partial g(\bx))\Big].
\end{equation}
\end{proposition}

\section{Recovery With Multiple Side Information}\label{RAMSI}
\subsection{Problem Statement}\label{problem}

\begin{figure}[t!]
\centering
\setlength{\tabcolsep}{1pt}
\renewcommand{\arraystretch}{0.1}
\subfigure[\vspace{-9pt}\texttt{View\#1}]{\includegraphics[width=0.10\textwidth]{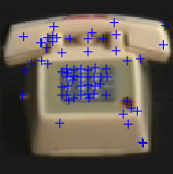}\label{figView1}}\hspace{5pt}
\subfigure[\vspace{-9pt}\texttt{View\#2}]{\includegraphics[width=0.10\textwidth]{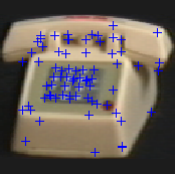}\label{figView2}}\hspace{5pt}
\subfigure[\texttt{View\#3}]{\includegraphics[width=0.10\textwidth]{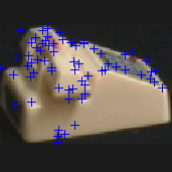}\label{figView3}}
\subfigure[\vspace{-9pt}1000-D histogram $\bx$ of \texttt{View\#1}]{\includegraphics[width=0.45\textwidth]{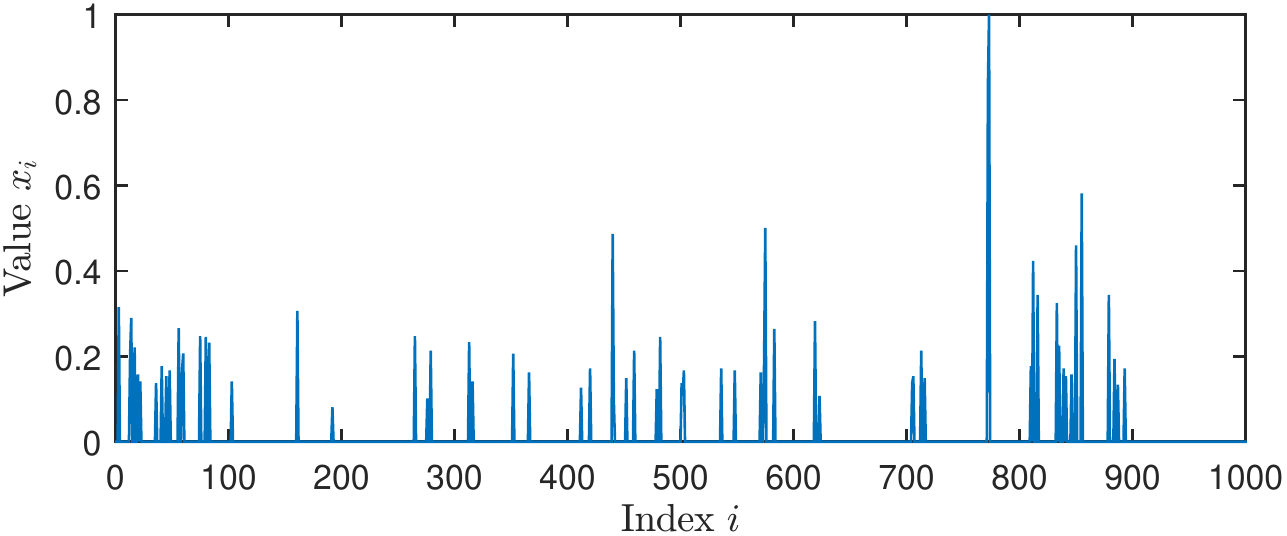}\label{figFeatVecX}}

\subfigure[\vspace{-9pt}1000-D histogram $\bz_{1}$ of \texttt{View\#2}]{\includegraphics[width=0.45\textwidth]{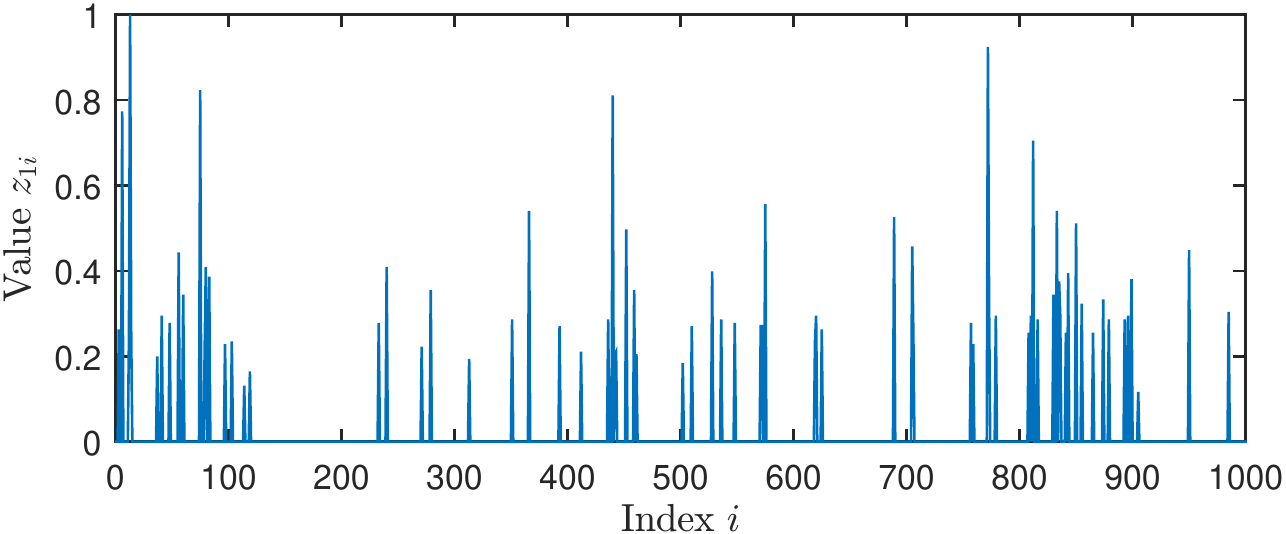}\label{figFeatVecZ1}}

\subfigure[1000-D histogram $\bz_{2}$ of \texttt{View\#3}]{\includegraphics[width=0.45\textwidth]{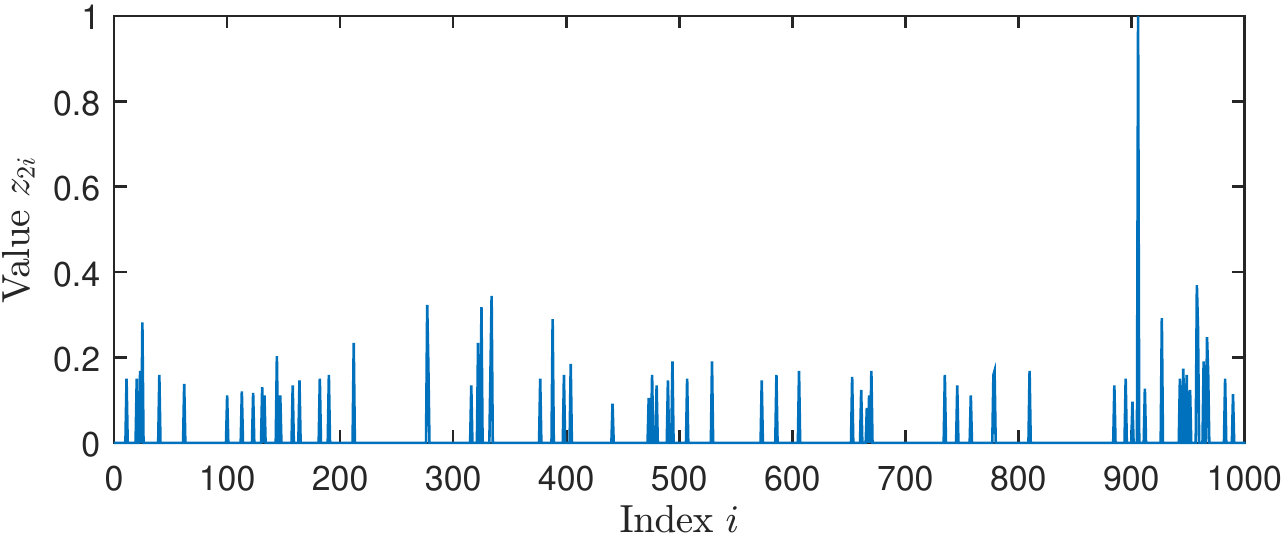}\label{figFeatVecZ2}}
\caption{SIFT-feature histograms of \texttt{Object} \texttt{\#60} in the \texttt{COIL-100} \cite{Nene96}.}\label{figFeatVec}
\end{figure}

\textcolor{black}{Let us consider the scenario of tiny cameras used for multiview object recognition. Fig. \ref{figFeatVec} shows three images from \texttt{View\#1}, \texttt{View\#2}, \texttt{View\#3} of \texttt{Object\#60} in the \texttt{COIL-100} database \cite{Nene96} and the corresponding histogram vectors of the SIFT-feature \cite{DLowe99} points. Each feature histogram is created by extracting all SIFT features from an image and then propagating down a hierarchical vocabulary tree based on a hierarchical $k$-means algorithm \cite{NisterCVPR06}.
It is evident (see Fig.~\ref{figFeatVec}) that the feature histogram acquired
from a given camera can be modeled as a sparse vector $\bx$.}

\textcolor{black}{In a practical distributed camera scenario, the dimensionality of the histograms can be high---in Fig.~\ref{figFeatVec} each histogram vector contains $1000$ elements. Classical CS methods can be deployed to reduce the dimension of each feature histogram with the purpose to convey a low-dimensional measurement vector $\by$ at a central node. However, classical CS does not leverage the inter-camera correlations, namely, the correlation among the histograms from different views. For example, suppose that the decoder has access to~$\bz_1$ and $\bz_2$, which are feature histograms of neighboring views that are naturally correlated with~$\bx$ [see Figs.~\ref{figFeatVecX}, \ref{figFeatVecZ1}, \ref{figFeatVecZ2}]. The $\ell_1\text{-}\ell_1$ minimization framework in Problem~\eqref{l1-l1minimization} can be used to recover $\bx$ from $\by$ given either $\bz_{1}$ or $\bz_{2}$, that is, only one side information vector. Moreover, there may be a chance that the~$\ell_1\text{-}\ell_1$ minimization framework performs worse than the conventional $\ell_{1}$ minimization method because of a poor-quality side information signal.} To address these two limitations, we propose a new \textit{reconstruction algorithm with multiple side information} (RAMSI), which aims at automatically and optimally utilizing information from multiple side information signals. The input of RAMSI is the measurement vector~$\by=\mathbf{\Phi} \bx$ and $J$ given side information signals $\bz_{1},\dots,\bz_{J}\in\mathbb{R}^{n}$.
The objective function of RAMSI is constructed using an $n$-$\ell_{1}$-norm function in Problem \eqref{l1-general}, i.e,
\vspace{-4pt}
\begin{equation}\label{n-l1g}
g(\bx) \hspace{-2pt}= \lambda \sum_{j=0}^{J}\hspace{-2pt}\|\mathbf{W}_{j}(\bx\hspace{-2pt}-\hspace{-2pt}\bz_{j})\|_{1},
    \vspace{-2pt}
\end{equation}
where $\bz_{0}=\mathbf{0}$ and $\mathbf{W}_{j}$ are diagonal weight matrices, $\mathbf{W}_{j}=\mathrm{diag}(w_{j1},w_{j2},...,w_{jn})$, wherein $w_{ji}\hspace{-2pt}>\hspace{-2pt}0$ is the weight in $\mathbf{W}_{j}$ at index $i$. Namely,
the objective function of the proposed weighted $n$-$\ell_{1}$ minimization problem is given by
\vspace{-5pt}
\begin{equation}\label{n-l1minimization}
\min_{\bx}\Big\{H(\bx)=\frac{1}{2}\|\mathbf{\Phi}\bx-\by\|^{2}_{2}+ \lambda \sum\limits_{j=0}^{J}\|\mathbf{W}_{j}(\bx-\bz_{j})\|_{1}\Big\}.
\end{equation}

\subsection{The Proposed RAMSI Algorithm}\label{theProposedRAMSI}
An important question arises when trying to solve the weighted $n$-$\ell_{1}$ minimization problem in \eqref{n-l1minimization}: How can one determine the weight values to improve the reconstruction by effectively leveraging the multiple side information signals? This also calls for a method that avoids recovery performance degradation when the quality of one or more side information signals is poor, namely, their correlation with the source signal of interest decreases. \textcolor{black}{As such, unlike prior studies \cite{MotaGLOBALSIP14,MotaARXIV14,MotaICASSP15}, our method needs to distribute relevant weights across \textit{multiple side information signals}}. We propose to solve the problem in \eqref{n-l1minimization} based on the proximal gradient method \cite{Beck09}, that is, at every iteration $k$ we iterate over a weight (i.e., $\mathbf{W}_{j}$) update step and a data (i.e., $\bx$) computation step.

\textcolor{black}{With the purpose to determine the weights $\{w_{ji}\}$---with $j\in\{0,\dots, J\}$ indexing the side information signal and $i\in\{1,\dots,n\}$ iterating over the data elements---we minimize the objective function $H(\bx)$ in \eqref{n-l1minimization} by considering $\bx$ fixed. To normalize the contribution of each side information signal during the iterative process, we impose a constraint
on the weights. We may have different strategies to update the weights $\{w_{ji}\}$; in this work, we use the constraint $\sum_{j=0}^{J}\mathbf{W}_{j}=\mathbf{I}_n$, where $\mathbf{I}_n$ is the identity matrix with dimensions $n\times n$. Hence, we compute $\{w_{ji}\}$ by separately optimizing Problem \eqref{n-l1minimization} at a given index $i$ of $\mathbf{W}_{j}$ as
\begin{equation}\label{n-l1-weight-minimization}
\min_{\{w_{ji}\}}\{H(\bx)\}= \min_{\{w_{ji}\}}\Big\{\lambda \sum\limits_{j=0}^{J}w_{ji}|x_{i}-z_{ji}|\Big\},
\end{equation}
where $z_{ji}$ is the element of $\bz_{j}$ at index $i$. Using the Cauchy-Schwartz
inequality~\cite{bachmann2012fourier}}, we achieve the minimization of \eqref{n-l1-weight-minimization} when all $w_{ji}|x_{i}-z_{ji}|$ are equal to a positive parameter $\eta_{i}$, that is, $w_{ji}=\eta_{i}/|x_{i}-z_{ji}|$. To ensure that zero-valued components of $|x_{i}-z_{ji}|$ do not prohibit the iterative computation of $w_{ji}$, we add a small parameter $\epsilon>0$ to $|x_{i}-z_{ji}|$; thus we set
\begin{equation}\label{CSl1-weights}
w_{ji}=\frac{\eta_{i}}{|x_{i}-z_{ji}|+\epsilon},
\end{equation}
Under the constraint $\sum_{j=0}^{J}\mathbf{W}_{j}=\mathbf{I}_n$ we obtain that
\begin{equation}\label{csEta}
\eta_{i} =\Big(\sum\limits_{j=0}^{J}\frac{1}{|x_{i}-z_{ji}|+\epsilon}\Big)^{-1}.
\end{equation}
Using \eqref{csEta} we can rewrite each weight $w_{ji}$ as
\begin{equation}\label{n-l1-weights}
w_{ji} = \frac{1}{1+(|x_{i}-z_{ji}|+\epsilon)\Big(\sum\limits_{l=0,l\neq j}^{J}(|x_{i}-z_{li}|+\epsilon)^{-1}\Big)}.
\end{equation}
With $\mathbf{W}_{j}$ fixed, RAMSI subsequently computes $\bx^{(k)}$ at iteration $k$ using \eqref{l1-proximal}, where the proximal operator $\Gamma_{\hspace{-2pt}\frac{1}{L}g}(\hspace{-1pt}x_{i}\hspace{-1pt})$ is computed in the following proposition.
\begin{proposition}\label{propMSI}
 The proximal operator $\Gamma_{\hspace{-2pt}\frac{1}{L}g}(\bx)$ in \eqref{l1-proximalOperator} for the problem of signal recovery with multiple side information, for which $g(\bx)=\lambda \sum_{j=0}^{J}\|\mathbf{W}_{j}(\bx-\bz_{j})\|_{1}$, is given by
\begin{equation}\label{n-l1-proximalOperatorElementFinal}
\Gamma_{\frac{1}{L}g}(x_{i}) = \left\{
\begin{array}{l}
x_{i}-\frac{\lambda}{L}  \sum\limits_{j=0}^{J}w_{ji}(-1)^{\mathfrak{b}(l<j)},\quad\mathrm{~if}\:\:\eqref{n-l1-proximalOperatorElementFinal:1} \\
 z_{li},~~~~~~~~~~~~~~~~~~~~~~~~~~~~~~~\mathrm{if}\:\:\eqref{n-l1-proximalOperatorElementFinal:2}
  \end{array}
\right.
\end{equation}
where
\begin{subequations}
\label{eq:Parent}
\begin{align}
&z_{li}\hspace{-0.5pt}+\hspace{-0.5pt}\frac{\lambda}{L}\hspace{-1pt} \sum\limits_{j=0}^{J}\hspace{-1pt}w_{ji}(\hspace{-1pt}-1\hspace{-1pt})^{\mathfrak{b}(l<j)}\hspace{-1pt}<\hspace{-1pt}x_{i}\hspace{-1pt}<\hspace{-1pt}z_{(l+1)i}\hspace{-1pt}+\hspace{-1pt}\frac{\lambda}{L}\hspace{-1pt} \sum\limits_{j=0}^{J}\hspace{-1pt}w_{ji}(\hspace{-1pt}-1\hspace{-1pt})^{\mathfrak{b}(l<j)} \label{n-l1-proximalOperatorElementFinal:1}\\
&\!\!\!\!\text{and}\nonumber\\
  &z_{li}\hspace{-1pt}+\hspace{-1pt}\frac{\lambda}{L} \hspace{-1pt}\sum\limits_{j=0}^{J}\hspace{-1pt}w_{ji}(\hspace{-1pt}-1\hspace{-1pt})^{\mathfrak{b}(l-1<j)}\hspace{-1pt}\leq \hspace{-1pt}x_{i}\hspace{-1pt}\leq \hspace{-1pt} z_{li}\hspace{-1pt}+\hspace{-1pt}\frac{\lambda}{L} \hspace{-1pt}\sum\limits_{j=0}^{J}\hspace{-1pt}w_{ji}(\hspace{-1pt}-1\hspace{-1pt})^{\mathfrak{b}(l<j)}, \label{n-l1-proximalOperatorElementFinal:2}
  \end{align}
\end{subequations}
and where, without loss of generality, we have assumed that $-\infty =z_{(-1)i}\leq z_{0i}\leq z_{1i}\leq\dots\hspace{-2pt}\leq z_{Ji}\leq z_{(J+1)i}=\infty$, and we have defined a boolean function
\begin{equation}\label{eq:bfunction}
\mathfrak{b}(l<j)=
\left\{
\begin{array}{l}
1,\quad \mathrm{if} \quad l<j \\
0, \quad\mathrm{otherwise}.
\end{array}
\right.
\end{equation}
\end{proposition}
with $l\in\{-1,\dots, J\}$.
\begin{proof}
The proof is given in Appendix~\ref{appendixProximal}.
\end{proof}
\vspace{2pt}
We sum up the proposed RAMSI in Algorithm \ref{ramsiAlg}, which is based on a fast iterative soft-thresholding algorithm (FISTA) algorithm \cite{Beck09}. The \textit{Stopping criterion} in Algorithm \ref{ramsiAlg} can be either a maximum iteration number $k_{\max}$, a relative variation of the objective function $H(\bx)$ in \eqref{n-l1minimization}, or a change of the number of nonzero components of the estimate $\bx^{(k)}$. In this work, the relative variation of $H(\bx)$ is chosen as stopping criterion.
\setlength{\textfloatsep}{0pt}
\begin{algorithm}[t!]
\DontPrintSemicolon \SetAlgoLined
\textbf{Input}: $\by,~\mathbf{\Phi},~\bz_{1},~\bz_{2},...,~\bz_{J}$;\\
\textbf{Output}: $\bhx$;\\
  \tcp{Initialization of variables and constants.}
    $\mathbf{W}_{0}^{(1)}\hspace{-1pt}=\hspace{-1pt}\mathbf{I}$; $\mathbf{W}_{j}^{(1)}\hspace{-1pt}=\hspace{-1pt}\mathbf{0}~(1\hspace{-1pt}\leq\hspace{-1pt} j\hspace{-1pt}\leq\hspace{-1pt} J)$;
   $\bu^{(1)}\hspace{-1pt}=\hspace{-1pt}\bx^{(0)}\hspace{-1pt}=\hspace{-1pt}\mathbf{0}$; $L\hspace{-1pt}=\hspace{-1pt}L_{\nabla f}$; $\lambda,\epsilon \hspace{-1pt}>\hspace{-1pt}0$; $t_{1}\hspace{-1pt}=\hspace{-1pt}1$; $k\hspace{-1pt}=\hspace{-1pt}0$; \\
  \While{Stopping criterion is false}{
       $k=k+1$; \\
       \tcp{Solving the solution given the weights.}
        $\nabla f(\bu^{(k)})=\mathbf{\Phi}^{\mathrm{T}}(\mathbf{\Phi} \bu^{(k)}-\by)$; \\
           $\bx^{(k)}= \Gamma_{\frac{1}{L}g}\Big(\bu^{(k)}-\frac{1}{L}\nabla f(\bu^{(k)})\Big)$; $\Gamma_{\frac{1}{L}g}(.)$ is given by \eqref{n-l1-proximalOperatorElementFinal};\\
           \tcp{Computing the updated weights.}
           $w_{ji}^{(k+1)} = \frac{1}{1+(|x_{i}^{(k)}-z_{ji}|+\epsilon)(\sum\limits_{l=0,\neq j}^{J}(|x_{i}^{(k)}-z_{li}|+\epsilon)^{-1})}$;
\\
       \tcp{Updating new values for the next iteration.}
        $t_{k+1}=(1+\sqrt{1+4t_{k}^{2}})/2$;\\
        \vspace{3pt}
        $\bu^{(k+1)}=\bx^{(k)}+\frac{t_{k}-1}{t_{k+1}}(\bx^{(k)}-\bx^{(k-1)})$;\\
  }
\Return $\bx^{(k)}$;
\caption{The proposed RAMSI algorithm.}\label{ramsiAlg}
\end{algorithm}

\section{\textcolor{black}{Measurement Bounds For Weighted $n\text{-}\ell_{1}$ Minimization }}
\label{ramsiBounds}
We now establish the measurement bounds for weighted $n\text{-}\ell_1$ minimization and analyze them in relation to the bounds for CS in \eqref{l1 bound} and the $\ell_{1}\text{-}\ell_{1}$ minimization method in \eqref{l1-l1 bound}. 
\subsection{Derived Measurement Bound}
\label{boundRAMSI-CS}
\subsubsection{Signal Traits}
\label{signalSetupRAMSI-CS}
\textcolor{black}{We begin our analysis by stating Assumptions \ref{asIndex}, \ref{asIndexZero}, \ref{asPartical} regarding the source $\bx$ with the support $s_{0}$ and the multiple side information signals~$\bz_{j}$, which will help us formalizing our bounds. Let $s_{j}$ denote the support of each difference vector $\bx\hspace{-2pt}-\hspace{-2pt}\bz_{j}$; namely,~$\|\bx-\bz_{j}\|_0=s_{j}$, where $j\in\{0,\dots, J\}$ and $\bz_0=\mathbf{0}$.}

In addition, without loss of generality, we make the following assumptions:
\textcolor{black}{
\begin{assumption}\label{asIndex}
\textcolor{black}{There are $p>0$ elements in each $\{\bx-\bz_{j}\}_{j=0}^{J}$ vector, the values of which are not equal to zero; meanwhile there are $n-q$, where $q<n$, elements in each $\{\bx-\bz_{j}\}_{j=0}^{J}$ vector, the values of which are equal to zero.}
\end{assumption}}
\textcolor{black}{By properly rearranging their elements in the same ordering for all difference vectors, we can represent the difference vectors as}
\begin{equation}\label{csSource}
\begin{array}{l}
  \hspace{-4pt}   \bx\hspace{-3pt}- \hspace{-3pt}\bz_{0}\hspace{-3pt}= \hspace{-3pt}(x_{1}~~~~~,     ...,x_{p}~~~~~~,      x_{p+1}~~~~~~~~~~,        ...,x_{q}~~~~~~,      0,...,0)\\
  \hspace{-4pt} \bx \hspace{-3pt}- \hspace{-3pt}\bz_{1}\hspace{-3pt}=\hspace{-3pt} (x_{1}\hspace{-3pt}-\hspace{-3pt}z_{11},...,x_{p}\hspace{-3pt}-\hspace{-3pt}z_{1p},x_{p+1}\hspace{-3pt}-\hspace{-3pt}z_{1(p+1)},...,x_{q}\hspace{-3pt}-\hspace{-3pt}z_{1q},0,...,0)\\
   \cdot\cdot\cdot\\
 \hspace{-4pt}     \bx \hspace{-3pt}- \hspace{-3pt}\bz_{ \hspace{-1pt}J}\hspace{-3pt}=\hspace{-3pt} (x_{1}\hspace{-3pt}-\hspace{-3pt}z_{\hspace{-1pt}J1},...,x_{p}\hspace{-3pt}-\hspace{-3pt}z_{\hspace{-1pt}Jp},x_{p+1}\hspace{-3pt}-\hspace{-3pt}z_{\hspace{-1pt}J(p+1)},...,x_{q}\hspace{-3pt}-\hspace{-3pt}z_{\hspace{-1pt}Jq},0,...,0).
\end{array}
\end{equation}
\textcolor{black}{Via Assumption~\ref{asIndex}, it is evident that $p\leq\min\{s_{j}\}$ and $q\geq\max\{ s_{j}\}$. As such, $p$ and $q$ can be seen as parameters that express the common part of the support and of the zero positions across the difference vectors. Furthermore, we assume that}
\textcolor{black}{\begin{assumption}\label{asIndexZero}
Following the rearrangement in~\eqref{csSource}, at a given position $i\in\{p+1,\dots, q\}$, we assume that $d_{i}\in\{1,\dots,J\}$ consecutive elements out of the $J+1$ elements $\{x_{i}-z_{ji}\}_{j=0}^{J}$ are zero. Namely, $\{x_{i}-z_{ji}\}_{j=l_{i}}^{l_{i}+d_{i}-1}=0$, with $l_i\in\{0,\dots,J\}$ being an auxiliary index indicating the start of the zero positions.
\end{assumption}}
\textcolor{black}{In other words, $d_i$ expresses the number of difference vectors~$\{\bx-\bz_{j}\}_{j=0}^{J}$ that have a zero element at the position indexed by $i$.}
Using Assumptions~\ref{asIndex} and~\ref{asIndexZero}, we can express the total number of zero elements in \eqref{csSource} as
\begin{equation}\label{csSparseEquality}
(J+1)(n-q)+\sum\limits_{i=p+1}^{q}d_{i}=(J+1)n-\sum\limits_{j=0}^{J}s_{j}.
\end{equation}

\textcolor{black}{
\begin{assumption}\label{asPartical}
\textcolor{black}{Let $i$ be a given position at the source and multiple side information vectors and, without loss of generality, assume that $-\infty =z_{(-1)i}\leq
z_{0i}\leq z_{1i}\leq\dots\hspace{-2pt}\leq z_{Ji}\leq z_{(J+1)i}=\infty$, where, for convenience, we introduced $z_{(-1)i}$ and $z_{(J+1)i}$ to denote $-\infty$ and $\infty$, respectively. Under Assumptions~\ref{asIndex} and~\ref{asIndexZero}, we assume that $x_{i}\in (z_{l_{i}i},z_{(l_{i}+1)i}]$, with $l_{i}\in\{-1,\dots, J\}$, or alias, $\mathrm{sign}(x_{i}-z_{ji})=(-1)^{\mathfrak{b}(l_{i}<j)}$, where $\mathfrak{b}(l_{i}<j)$ is defined in~\eqref{eq:bfunction}}.
\end{assumption}}

\subsubsection{The Measurement Bound}
These Assumptions  \ref{asIndex}, \ref{asIndexZero}, \ref{asPartical} help us to derive a bound on the number of measurements required by RAMSI to successfully recover the target signal. Our generic bound is defined by Theorem \ref{RAMSIBound}, while a simpler but looser bound is described in Section~\ref{boundAnalysis}.

\begin{theorem}[Measurement bound for compressed sensing with multiple side information]\label{RAMSIBound}
\textcolor{black}{The number of measurements $m_{n\text{-}\ell_{1}}$ required by weighted $n\text{-}\ell_{1}$ minimization to recover the signal $\bx$, given low-dimensional measurements $\by=\mathbf\Phi\bx$ and $J$ side information signals $\bz_{j}$, is bounded as
 \begin{equation}\label{upperBoundComputeCSApplyAEtaSetMeasurement}
     m_{n\text{-}\ell_{1}}\geq 2\bar{a}_{n\text{-}\ell_{1}} \log\frac{n}{\bar{s}_{n\text{-}\ell_{1}}}\hspace{-1pt}+\hspace{-1pt}
     \frac{7}{5}
     \bar{s}_{n\text{-}\ell_{1}}\hspace{-1pt}+\delta_{n\text{-}\ell_{1}}+1,
\end{equation}
where, under Assumptions  \ref{asIndex}, \ref{asIndexZero}, \ref{asPartical}, $\bar{a}_{n\text{-}\ell_{1}}$, $\bar{s}_{n\text{-}\ell_{1}}$, and $\delta_{n\text{-}\ell_{1}}$
are defined as
\begin{subequations}
\begin{align}
\bar{a}_{n\text{-}\ell_{1}}&=\sum\limits_{i=1}^{p}\hspace{-2pt}a_{i}^{2}\label{eq:hbarQuantity}\\
\bar{s}_{n\text{-}\ell_{1}}&=p+\hspace{-4pt}\sum\limits_{i=p+1}^{q}(1-c_{i})\label{eq:sbarnL1}\\
\delta_{n\text{-}\ell_{1}}&=\hspace{-1pt}(\kappa_{n\text{-}\ell_{1}}-1)(\bar{s}_{n\text{-}\ell_{1}}-p),\label{eq:DeltaNL1}
\end{align}
\end{subequations}
wherein $a_{i}=\sum\limits_{j=0}^{J} w_{ji}( -1)^{\mathfrak{b}(l_{i}<j)}$, $c_i=d_{i}\Big(\sum\limits_{j=0}^{J}\frac{\epsilon}{|x_{i}-z_{ji}|+\epsilon}\Big)^{-1}$, and
\begin{equation}\label{eq:KappaFinalDefinition}
\kappa_{n\text{-}\ell_{1}}= \frac{\sqrt{4}\cdot\min \{c_{i}\}}{\sqrt{\pi\log(n/\bar{s}_{n\text{-}\ell_{1}})}(2\cdot\min\{
c_{i}\}-1)}.
\end{equation}}
\end{theorem}
\textcolor{black}{\begin{proof}
The proof is given in Appendix~\ref{proof:ProofTheorem4}.
\end{proof}}

\subsection{Further Analysis and Comparison with Known Bounds}\label{boundAnalysis}

We now relate the derived bound in Theorem~\ref{RAMSIBound} with the bounds of compressed sensing and compressed sensing with prior information, which are reported in Sec. \ref{sparseRecovery}.
\begin{corollary}[Relations with the known bounds]\label{corBoundRelation}
The RAMSI bound $m_{n\text{-}\ell_{1}}$ in \eqref{upperBoundComputeCSApplyAEtaSetMeasurement} becomes
\begin{enumerate}[\upshape (a)]
  \item\label{cs-l1Relation} the $\ell_{1}$-minimization bound $m_{\ell_{1}}$ in \eqref{l1 bound}, when $\mathbf{W}_{0}=\mathbf{I}_n$ and $\mathbf{W}_{j}=\mathbf{0}$ for $j\in\{1,\dots,J\}$, that is,
\begin{equation}\label{RAMSI-l1 bound}
m_{n\text{-}\ell_{1}}\equiv m_{\ell_{1}}\geq 2s_{0}\log\frac{n}{s_{0}}+ \frac{7}{5}s_{0} + 1,
\end{equation}

\item\label{cs-l1-l1Relation} the $\ell_1\text{-}\ell_1$ minimization bound $m_{\ell_{1}\text{-}\ell_{1}}$ in \eqref{l1-l1 bound}, when $\mathbf{W}_{0}=\mathbf{W}_{1}=\frac{1}{2}\mathbf{I}_n$ and $\mathbf{W}_{j}=\mathbf{0}$ for $j\in\{2,\dots,J\}$, that is,
\begin{equation}\label{RAMSI-l1l1 bound}
    m_{n\text{-}\ell_{1}}\equiv m_{\ell_{1}\text{-}\ell_{1}} \geq 2\bar{h}\log\frac{n}{\bar{s}_{\ell_{1}\text{-}\ell_{1}}}
+\frac{7}{5}\bar{s}_{\ell_{1}\text{-}\ell_{1}}+1,
\end{equation}
where $\bar{h}$ is given by \eqref{l1-l1 sparse setHBar} and $\bar{s}_{\ell_{1}\text{-}\ell_{1}}=\frac{s_{0}+s_{1}}{2}$.
\end{enumerate}
\end{corollary}

\textcolor{black}{\begin{proof}
The proof is given in Appendix~\ref{proof:ProofTheorem4}.
\end{proof}}

We now theoretically evaluate our bound for weighted $n\text{-}\ell_{1}$ minimization with the $\ell_1\text{-}\ell_1$ minimization~\cite{MotaGLOBALSIP14,MotaARXIV14,MotaICASSP15} bound in \eqref{l1-l1 bound}. 
To this end, we derive a simple bound in \eqref{boundCSAppro}, which approximates our bound in \eqref{upperBoundComputeCSApplyAEtaSetMeasurement}. Furthermore, we introduce two looser bounds---one for each method---that are independent from the values of~$\bx,~\bz_{j}$.

\textbf{The simpler bound}. Our bound for weighted $n\text{-}\ell_{1}$ minimization in \eqref{upperBoundComputeCSApplyAEtaSetMeasurement} becomes approximately
\begin{equation}\label{boundCSAppro}
\widetilde{m}_{n\text{-}\ell_{1}}\geq 2\bar{a}_{n\text{-}\ell_{1}} \log\frac{n}{p}\hspace{-1pt}+\hspace{-1pt}\frac{7}{5}p+1.
\end{equation}
\textcolor{black}{Our approximation has the following reasoning: Firstly, Lemma
\ref{Proposition delta} in Appendix \ref{ap:BoundApprox} proves that the $\delta_{n\text{-}\ell_{1}}$
term in \eqref{upperBoundComputeCSApplyAEtaSetMeasurement} is negative; hence, the bound in~\eqref{boundCSAppro} is looser. Secondly, since $\epsilon>0$ is very small, we have $c_{i}\rightarrow\ 1^-$ [see Lemma \ref{Proposition
delta} for a quick explanation]. Consequently, $\bar{s}_{n\text{-}\ell_{1}}\approx p$ and $\delta_{n\text{-}\ell_{1}}\approx
0$ from~\eqref{eq:sbarnL1} and~\eqref{eq:DeltaNL1}, respectively; thereby, leading to our approximation. This simple bound is easier to evaluate compared to the bound in \eqref{upperBoundComputeCSApplyAEtaSetMeasurement},
as we only need to compute $\bar{a}_{n\text{-}\ell_{1}}$
and $p$. Furthermore, according to \eqref{eq:hbarQuantity} and Assumption \ref{asIndex}, we can write that $\bar{a}_{n\text{-}\ell_{1}}\leq p\leq\min\{s_{j}\}$, that is, $p$ in \eqref{boundCSAppro} is smaller
than $s_{0}$ in \eqref{l1 bound} and $\bar{s}_{\ell_{1}\text{-}\ell_{1}}$
in \eqref{RAMSI-l1l1 bound}. Consequently, the simple bound in
\eqref{boundCSAppro} for weighted $n\text{-}\ell_1$ minimization is sharper than the $\ell_1$ minimization and the $\ell_1\text{-}\ell_1$ minimization bounds, which means that the bound in~\eqref{upperBoundComputeCSApplyAEtaSetMeasurement} is even sharper.}

\textbf{The looser bounds}. Weighted $n\text{-}\ell_1$ minimization and $\ell_1\text{-}\ell_1$ minimization
have looser bounds that are independent from the values of $\bx,~\bz_{j}$, as given by
\begin{equation}\label{boundCSApproLoosest}
\widehat{m}_{n\text{-}\ell_{1}}\geq 2p\log\frac{n}{p}+\frac{7}{5}p+1,
\end{equation}
\begin{equation}\label{boundL1L1Loosest}
\widehat{m}_{\ell_{1}\text{-}\ell_{1}}\geq 2\rho\log\frac{n}{\bar{s}_{\ell_{1}\text{-}\ell_{1}}}
+\frac{7}{5}\bar{s}_{\ell_{1}\text{-}\ell_{1}}+1,
\end{equation}
where $p$ is defined in Assumption~\ref{asIndex}, $\rho=\min\{s_{0},s_{1}\}$, and $\bar{s}_{\ell_{1}\text{-}\ell_{1}}=\frac{s_{0}+s_{1}}{2}$.
\textcolor{black}{We obtain these bounds as follows:\ The bounds in \eqref{boundCSAppro} and \eqref{RAMSI-l1l1 bound} depend on the values of $\bx$ and $\bz_{j}$, with $j\in\{1,\dots,J\}$, via the quantities $\bar{a}_{n\text{-}\ell_{1}}$ and $\bar{h}$, respectively. From \eqref{eq:hbarQuantity} and \eqref{l1-l1 sparse setHBar}, we observe that $\bar{a}_{n\text{-}\ell_{1}}\leq p$ and $\bar{h}\leq \min\{s_{0},s_{1}\}$. Hence, by replacing $\bar{a}_{n\text{-}\ell_{1}}$
and $\bar{h}$ with their maximum value leads to the looser bounds in \eqref{boundCSApproLoosest} and in \eqref{boundL1L1Loosest}.}

\textcolor{black}{The bounds in \eqref{boundCSApproLoosest} and \eqref{boundL1L1Loosest} reveal the advantage of using multiple side information signals in compressed sensing: It is evident that, by definition, $p\leq\rho$; hence, $\widehat{m}_{n\text{-}\ell_{1}}\leq\widehat{m}_{\ell_{1}\text{-}\ell_{1}}$. Moreover, the higher the number of side information signals the smaller the bound is expected to become (because $p\leq\min\{s_{j}\}$).}

\textcolor{black}{Finally, according to \eqref{boundL1L1Loosest}, if the side information signal $\bz_{1}$ is not good enough, i.e., if $s_{1}\gg s_{0}$, then $\widehat{m}_{\ell_{1}\text{-}\ell_{1}}>m_{\ell_{1}}$ because of $\bar{s}_{\ell_{1}\hspace{-1pt}\text{-}\ell_{1}}\hspace{-2pt}\gg\hspace{-2pt}s_{0}$, thereby highlighting the limitations of the $\ell_{1}\text{-}\ell_{1}$ minimization method compared to the proposed weighted $n\text{-}\ell_{1}$ minimization
approach.}

\section{Experimental Results}
\label{Experiment}
We present numerical experiments that demonstrate the established bounds and the performance of the RAMSI algorithm on sparse signals with different characteristics. We also analyze how the quality of the side information signals impacts the number of measurements required by RAMSI to recover the target signal. In addition, we evaluate RAMSI using real correlated feature histogram vectors, extracted from a multiview image database \cite{Nene96}.
\subsection{Experimental Setup}\label{numericalExperiments}
We consider the reconstruction of a generated sparse source $\bx$ given known side information signals $\bz_{j}$, $j\in\{1,2,3\}$. We generate $\bx$ with $n \hspace{0pt}= \hspace{0pt}1000$, and with support size $s_{0}\hspace{0pt}=\hspace{0pt}128$, from the i.i.d. zero-mean, unit-variance Gaussian distribution. Firstly, we consider the scenario where the side information signals $\bz_{j}$ are well-correlated to the source $\bx$ leading to a small number of nonzero elements in the vectors ${\bx\hspace{0pt}-\hspace{0pt}\bz_{j}}$. In this scenario, the side information signals, $\bz_{j}$, $j\in\{1,2,3\}$, are generated satisfying $s_{j}=\|\bx-\bz_{j}\|_{0}=64$. Moreover, similar to the setup in \cite{MotaGLOBALSIP14,MotaARXIV14}, a parameter is controlling the number of positions of nonzeros for which both $\bx$ and $\bx\hspace{0pt}-\hspace{0pt}\bz_{j}$ coincide. Let us denote this number by $r_{j}$. For instance, if $r_{j}\hspace{0pt}=\hspace{0pt}51$, $\bx$ has 51 nonzero positions that coincide with 51 nonzero positions of $\bx\hspace{0pt}-\hspace{0pt}\bz_{j}$. This incurs a significant error between the source and the side information, which is given by $\|\bz_{j}\hspace{0pt}-\hspace{0pt}\bx\|_{2}/\|\bx\|_{2}\hspace{0pt} \approx \hspace{0pt}0.56$.

To assess the performance of the algorithm and the bounds when the quality of the side information is poor, we generate side information signals that are not well-correlated to the source $\bx$. Their poor qualities are expressed via higher values of $s_{j}$; for example, $s_{j}\hspace{0pt}=\hspace{0pt}256$ and $s_{j}\hspace{0pt}=\hspace{0pt}352$. Furthermore, we set $r_{j}\hspace{0pt}=\hspace{0pt}128$, namely, 128 nonzero positions of $\bx$ coincide with 128 nonzero positions out of the total 256 or 352 nonzero positions of the side information signals, $\bz_{j}$, $j\in\{1,2,3\}$. This leads to very high errors, e.g., $\|\bz_{j}\hspace{0pt}-\hspace{0pt}\bx\|_{2}/\|\bx\|_{2}\hspace{0pt} \approx \hspace{0pt}1.12$ for $s_{j}\hspace{0pt}=\hspace{0pt}256$, and the supports $s_{j}$ of $\bx\hspace{0pt}-\hspace{0pt}\bz_{j}$ are much higher than that of $\bx$. To constrain the number of cases, we set all $s_{j}$ equal.

Furthermore, we evaluate the algorithm using real experimental data, namely, feature histograms extracted from a multiview image database. Given an image, its feature histogram is formed as in Section \ref{problem}. The size of the vocabulary tree \cite{NisterCVPR06} depends on the value of $k$ and the number of hierarchies, for example, when $k\hspace{0pt} = \hspace{0pt}10$ and 3 hierarchies are used leading to $n \hspace{0pt}= \hspace{0pt}1000$. Because of the small number of features in a single image, the histogram vector $\bx$ is indeed sparse. In our setup, $\bx$ is first projected into the compressed vector $\by$ which is to be sent to the decoder. At the joint decoder, we recover $\bx$ given already decoded histograms of neighbor views, e.g., $\bz_{j}$, $j\in\{1,2,3\}$. In this work, we use the \texttt{COIL-100} dataset \cite{Nene96} containing multiview images of 100 small objects with different angle degrees. To ensure that our experimental setup reflects a realistic scenario, we randomly select 4 neighboring views from \texttt{Object} \texttt{\#16} [Fig. \ref{figObj16}] and \texttt{Object} \texttt{\#60} [Fig. \ref{figObj60}] in the \texttt{COIL-100}~\cite{Nene96} multiview database. Specifically, the four neighboring views are assigned as $\bz_{1}$, $\bx$, $\bz_{2}$, $\bz_{3}$, respectively, of which the third side information $\bz_{3}$ is set to the furthest neighbor of the source $\bx$.

\subsection{Performance Evaluation}\label{performance}
\label{Multi-ViewRecognition}

\subsubsection{Synthetic Signal Reconstruction}\label{generatedPer}
\begin{figure}[t!]
  \centering
\subfigure[$s_{0}\hspace{0pt}=\hspace{0pt}128$, $s_{1}\hspace{0pt}=\hspace{0pt}s_{2}\hspace{0pt}=\hspace{0pt}s_{3}\hspace{0pt}=\hspace{0pt}64$]{\includegraphics[width=0.48\textwidth]{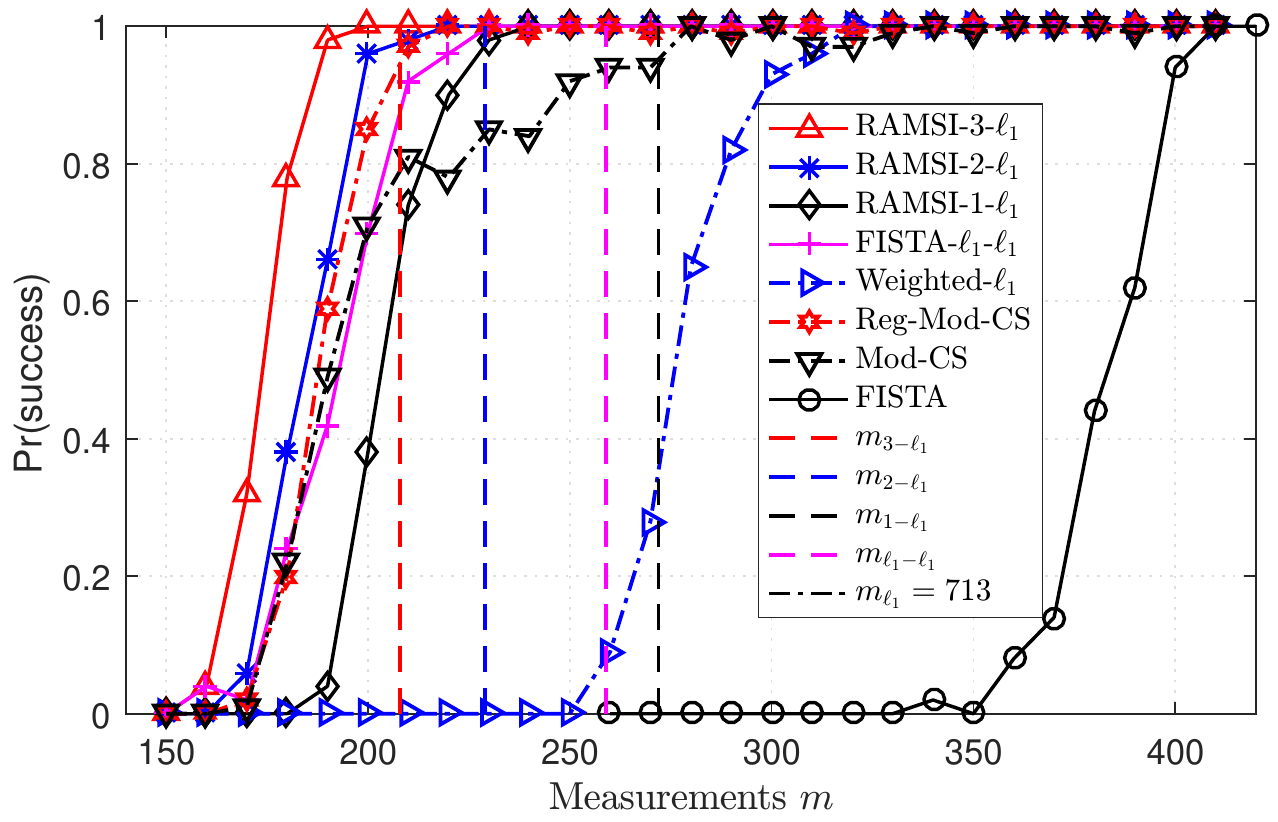}\label{figS64}} \subfigure[$s_{0}\hspace{0pt}=\hspace{0pt}128$, $s_{1}\hspace{0pt}=\hspace{0pt}s_{2}\hspace{0pt}=\hspace{0pt}s_{3}\hspace{0pt}=\hspace{0pt}256$]{\includegraphics[width=0.48\textwidth]{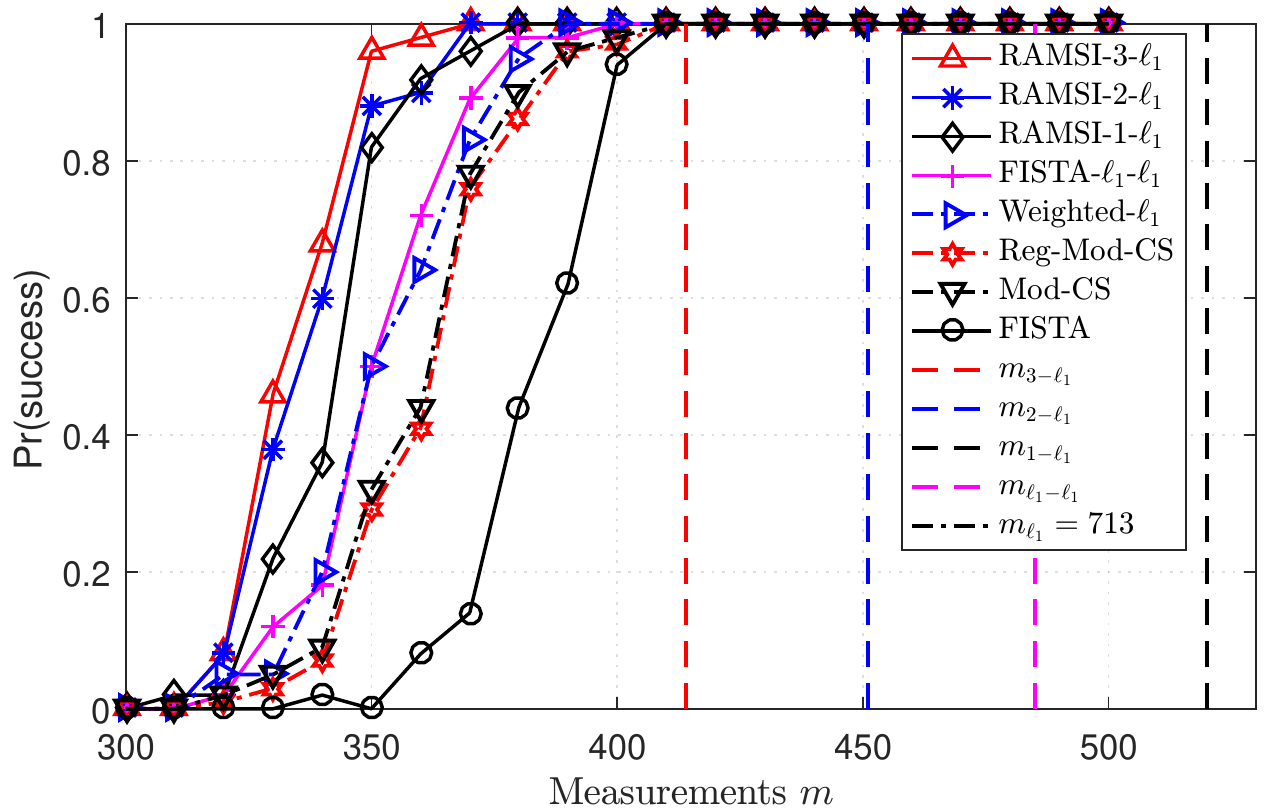}\label{figS256}} 
\subfigure[$s_{0}\hspace{0pt}=\hspace{0pt}128$, $s_{1}\hspace{0pt}=\hspace{0pt}s_{2}\hspace{0pt}=\hspace{0pt}s_{3}\hspace{0pt}=\hspace{0pt}352$]{\includegraphics[width=0.48\textwidth]{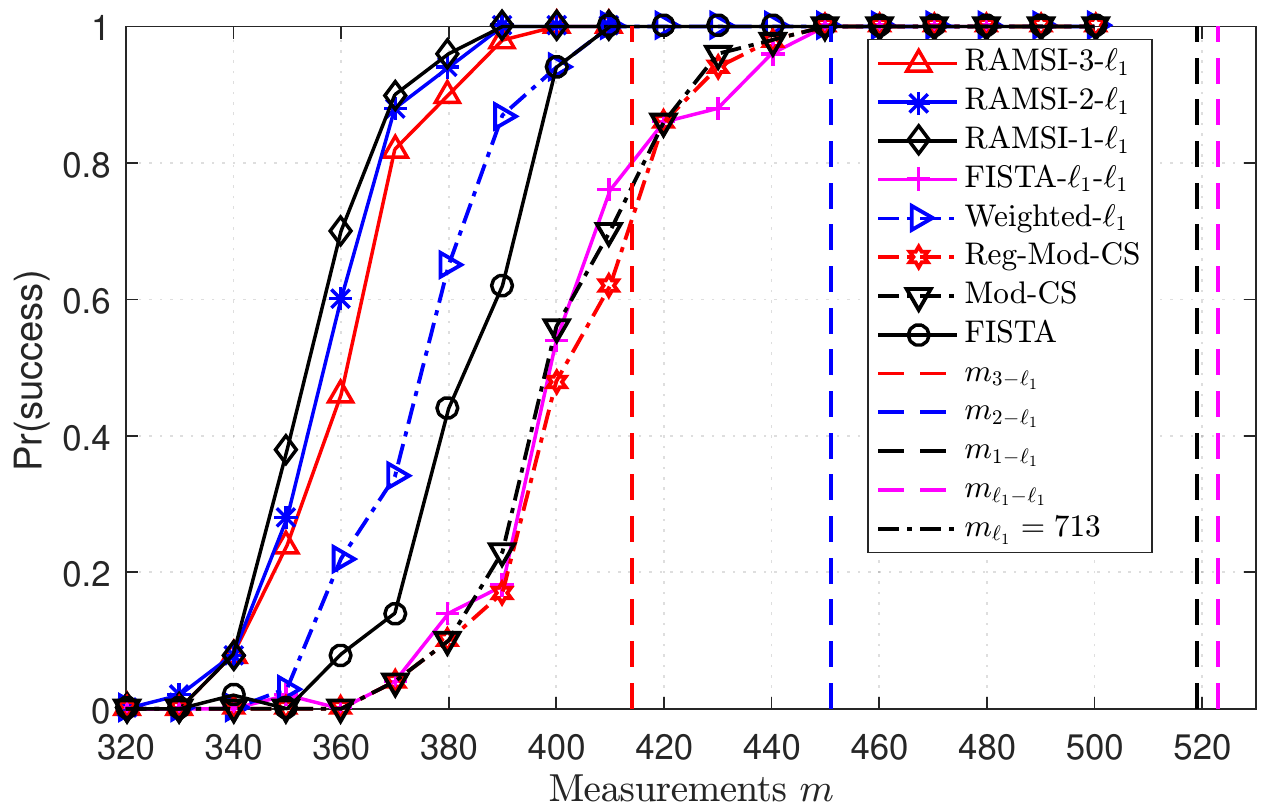}\label{figS384}} 
  \caption{Successful probabilities and measurement bounds of the original 1000-D $\bx$ vs. number of measurements $m$ for RAMSI using one, two, three side information signals.}
  \label{figAccuracy}
\end{figure}
We now evaluate the obtained bounds for weighted $n\text{-}\ell_{1}$ minimization [see \eqref{upperBoundComputeCSApplyAEtaSetMeasurement}] against the bounds for classical CS [see \eqref{l1 bound}] and $\ell_{1}\text{-}\ell_{1}$ minimization [see \eqref{l1-l1 bound}]. Furthermore, we assess the performance of RAMSI against existing state-of-the-art approaches. We consider the probability of successful recovery, denoted as $\mathrm{Pr(success)}$; for a fixed number of measurements $m$, $\mathrm{Pr(success)}$ is the number of times, in which the source $\bx$ is recovered as $\bhx$ with an error $\|\bhx\hspace{-1pt}-\hspace{-1pt}\bx\|_{2}/\|\bx\|_{2}\hspace{0pt} \leq \hspace{0pt}10^{-2}$, divided by the total number of \textcolor{black}{100 Monte-Carlo iterations, alias trials, where each trial considers different generated quantities $\bx,~\bz_{1},~\bz_{2},~\bz_{3},~\mathbf{\Phi}$.}

Let RAMSI-$J$-$\ell_{1}$ denote the RAMSI algorithm that uses $J$ side information signals, that is, one ($\bz_{1}$), two ($\bz_{1}$, $\bz_{2}$), or three ($\bz_{1}$, $\bz_{2}$, $\bz_{3}$) side information signals, respectively. Moreover, let FISTA-$\ell_{1}$-$\ell_{1}$ denote the $\ell_{1}$-$\ell_{1}$ method that uses only one side information (that is, $\bz_{1}$) \cite{MotaGLOBALSIP14}. Furthermore, the existing FISTA \cite{Beck09}, FISTA-$\ell_{1}$-$\ell_{1}$ \cite{MotaGLOBALSIP14}, \textcolor{black}{Mod-CS \cite{NVaswaniTSP10}, Reg-Mod-CS \cite{NVaswaniTSP12}, and weighted-$\ell_{1}$ \cite{KhajehnejadXAH11}} algorithms are used for comparison. \textcolor{black}{The support estimate of $\bx$ that is considered as prior information by the Mod-CS, Reg-Mod-CS, and Weighted-$\ell_{1}$ algorithms is given by the support of one of the side information signals, that is, $\bz_{1}$}. Let $m_{3\text{-}\ell_{1}}$, $m_{2\text{-}\ell_{1}}$, $m_{1\text{-}\ell_{1}}$, $m_{\ell_{1}\hspace{-1pt}\text{-}\ell_{1}}$, and $m_{\ell_{1}}$ denote the corresponding bounds of RAMSI-$3$-$\ell_{1}$, RAMSI-$2$-$\ell_{1}$, RAMSI-$1$-$\ell_{1}$, FISTA-$\ell_{1}$-$\ell_{1}$, and FISTA. The source $\bx$ is compressed into different lower-dimensions $\by$ and we assess the bounds and the accuracy of a reconstructed $\bhx$ versus the number of measurements $m$. For RAMSI we set $\epsilon\hspace{0pt}=\hspace{0pt}10^{-5}$, $\lambda\hspace{0pt}=\hspace{0pt}10^{-5}$.

Fig. \ref{figS64} depicts the bounds as well as the probabilities of successful recovery versus the number of measurements for $s_{j}=64$. The results show clearly that RAMSI-3-$\ell_{1}$ gives the sharpest bound and the highest reconstruction performance. Furthermore, the performance of RAMSI-2-$\ell_{1}$ is higher than those of RAMSI-1-$\ell_{1}$ and FISTA-$\ell_{1}$-$\ell_{1}$. In this scenario, where the side information is of high quality, FISTA-$\ell_{1}$-$\ell_{1}$ outperforms our RAMSI-1-$\ell_{1}$ method. Furthermore, \textcolor{black}{Reg-Mod-CS outperforms RAMSI-1-$\ell_{1}$, FISTA-$\ell_{1}$-$\ell_{1}$, Mod-CS, and Weighted-$\ell_{1}$. In addition, the weighted-$\ell_{1}$ method performs better than FISTA but worse than the remaining schemes.} Hence, in this case, the use of equal weights---as in FISTA-$\ell_{1}$-$\ell_{1}$---leads to a higher performance than using adaptive weights with only one side information signal, as it is the case for RAMSI-1-$\ell_{1}$. This can be explained by comparing the $m_{1\text{-}\ell_{1}}$ and the $m_{\ell_{1}\hspace{-1pt}\text{-}\ell_{1}}$ bounds: It is clear that $\bar{a}_{1\text{-}\ell_{1}}$ [see \eqref{eq:hbarQuantity}] is greater than $\bar{h}$ [see \eqref{l1-l1 sparse setHBar}]. By combining this observation with the small value of $s_{1}$---due to the good-quality side information signal $\bz_{1}$---, which in turn results in a small $\bar{s}_{{\ell_{1}\text{-}\ell_{1}}}$ value, explains why the bound as well as the recovery performance of FISTA-$\ell_{1}$-$\ell_{1}$ are better than those of RAMSI-1-$\ell_{1}$ [see magenta and black lines in Fig. \ref{figS64}]. We can conclude that by exploiting multiple side information signals we can obtain the best performance and when dealing with only one good-quality side information signal we may choose to reconstruct using equal weights.

Figs. \ref{figS256} and \ref{figS384} present the bounds and the reconstruction performance versus the number of measurements when the side information signals are less correlated with the signal of interest, that is, $s_{j}\hspace{0pt}=\hspace{0pt}256$ and $s_{j}\hspace{0pt}=\hspace{0pt}352$. In this scenario, all RAMSI configurations outperform the FISTA, FISTA-$\ell_{1}$-$\ell_{1}$, \textcolor{black}{Weighted-$\ell_{1}$, Reg-Mod-CS, and Mod-CS methods}. The performance of RAMSI-1-$\ell_{1}$ is better than those of FISTA-$\ell_{1}$-$\ell_{1}$, \textcolor{black}{Weighted-$\ell_{1}$, Reg-Mod-CS, and Mod-CS}, where the same side information is used (that is, $\bz_{1}$). Interestingly, in Fig. \ref{figS384}, we observe that the accuracy of FISTA-$\ell_{1}$-$\ell_{1}$ is worse than that of FISTA, i.e., the side information $\bz_{1}$ does not help, however, RAMSI-1-$\ell_{1}$ still outperforms FISTA. These results highlight the drawback of the $\ell_{1}$-$\ell_{1}$ minimization method when the side information is of poor quality. Although encountering poor side information signals, all the RAMSI versions achieve better results than FISTA due to the proposed re-weighted $n$-$\ell_{1}$ minimization algorithm. We observe that the performance results of RAMSI-2-$\ell_{1}$ and RAMSI-3-$\ell_{1}$ are slightly worse than that of RAMSI-1-$\ell_{1}$. These small penalties can be explained by the poor quality of the side information signals, which have an impact on the iterative update process.

\subsubsection{Side Information Quality-Dependence Analysis}\label{SIPer}
We now assess the impact of the side information quality on the number of measurements required to successfully reconstruct the target signal. We consider various cases for the side information quality, expressed through different values for $s_{1}=s_{2}=s_{3}$, ranging from 20 to 400 (where $s_j$ is the support size of the difference vector $\bx-\bz_j$). The source $\bx$ has $s_{0}\hspace{0pt}=\hspace{0pt}128$ and it is generated as in Sec. \ref{numericalExperiments}.
For a fixed value of $s_{1}\hspace{0pt}=\hspace{0pt}s_{2}\hspace{0pt}=\hspace{0pt}s_{3}$, we report the number of measurements required by RAMSI-$3$-$\ell_{1}$, RAMSI-$2$-$\ell_{1}$,
RAMSI-$1$-$\ell_{1}$, FISTA-$\ell_{1}$-$\ell_{1}$, FISTA to achieve a probability of successful recovery bounded as $\mathrm{Pr(success)}\hspace{0pt}\geq \hspace{0pt}0.98$.

\begin{figure}[tp!]
\centering
\includegraphics[width=0.47\textwidth]{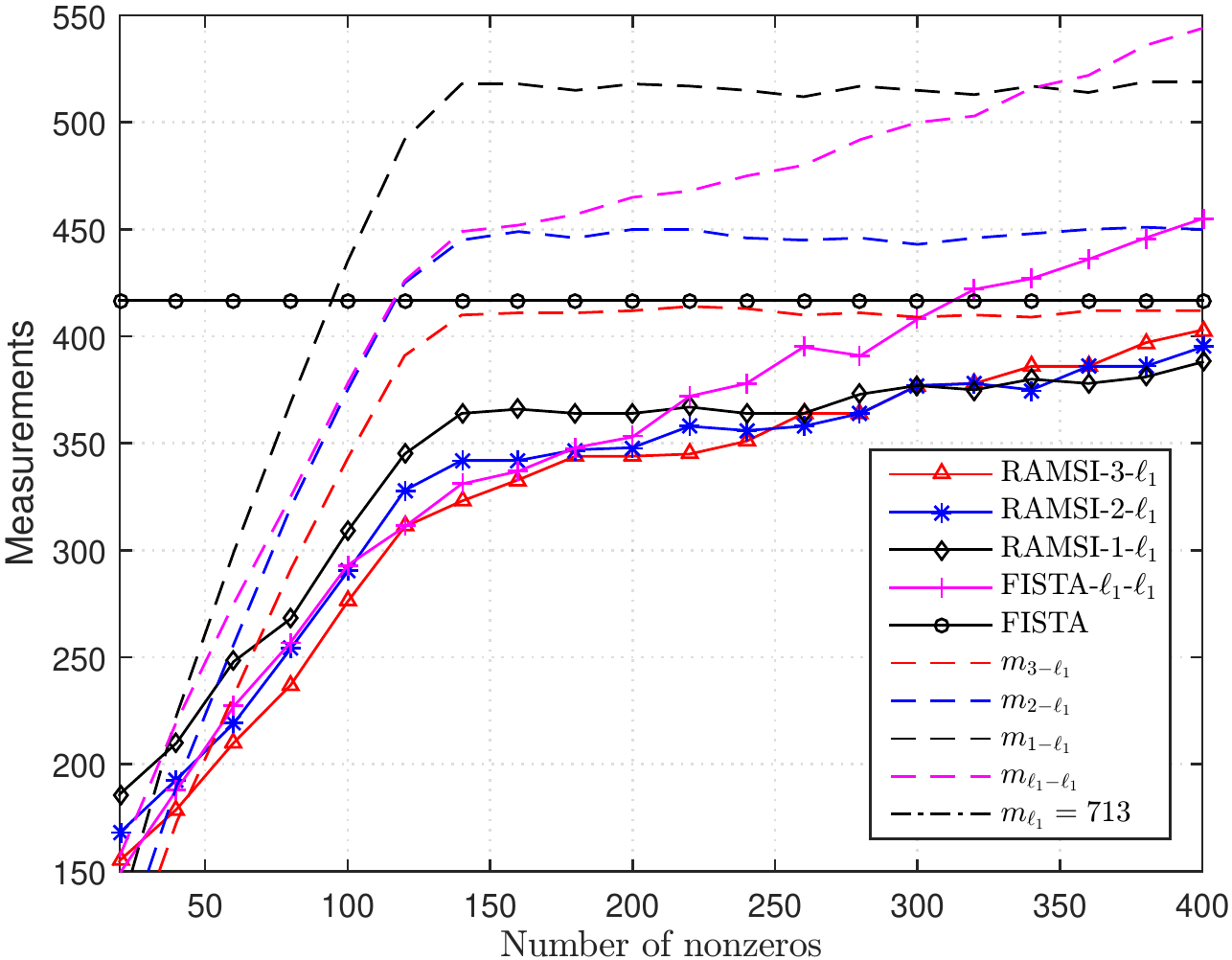}
\caption{Number of measurements vs. number of nonzeros $s_{1}\hspace{0pt}=\hspace{0pt}s_{2}\hspace{0pt}=\hspace{0pt}s_{3}$ of side information signals given the source $s_{0}\hspace{0pt}=\hspace{0pt}128$ for RAMSI using one, two, three side information signals.}
\label{figSIQuality}
\end{figure}

Fig. \ref{figSIQuality} shows that the recovery performance is clearly improved when leveraging side information signals; and the more signals are considered the better the performance. However, when the support size of the difference vectors $\bx-\bz_j$ meets that of the target signal $s_0$, the use of side information does not improve the performance anymore. When $s_{j}\hspace{0pt}>\hspace{0pt}s_{0}$, the performance and the measurement bound of FISTA-$\ell_{1}$-$\ell_{1}$ are increasing with the value of $s_{1}$. Specifically, when $s_{1}\hspace{0pt}>\hspace{0pt}315$, the performance of FISTA-$\ell_{1}$-$\ell_{1}$ is worse than that of FISTA, thereby illustrating again the limitations of the $\ell_{1}$-$\ell_{1}$ minimization method when the side information is of poor quality.
As shown in Fig. \ref{figSIQuality}, the performance of RAMSI---both in terms of the theoretical bounds and the practical results---is robust against poor-quality side information signals. The theoretical bounds are sharper when the number of side information signals increases and remain approximately constant when $s_{j}$ increases after a threshold (indicating increasing-poor side information quality). When $s_{j}\hspace{0pt}>\hspace{0pt}300$, the number of measurements of RAMSI-3-$\ell_{1}$ and RAMSI-2-$\ell_{1}$ are slightly worse than those of RAMSI-1-$\ell_{1}$, approaching the number of measurements of FISTA [this behavior was also observed in Fig. \ref{figS384}]. To address this issue, we can adaptively select the best performance among the RAMSI configurations. For instance, when one side information is available, we can choose to use equal weights in the case of good side information. Furthermore, we can ensure that RAMSI's performance is not worse than FISTA by weighing dominantly on the source rather than on the poor-quality side information signals. To do so, we need to make an estimate of the quality of the side information signals, which is left as a topic for future research.

\vspace{-0pt}
\subsubsection{Feature Histogram Reconstruction}\label{featPer}
\begin{figure}[tp!]
\centering
\setlength{\tabcolsep}{1pt}
\renewcommand{\arraystretch}{0.1}
\subfigure[\vspace{-5pt}\texttt{Object}
\texttt{\#16}]{\includegraphics[width=0.12\textwidth]{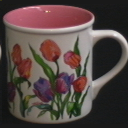}\label{figObj16}} \hspace{25pt} \subfigure[\vspace{-5pt}\texttt{Object} \texttt{\#60}]{\includegraphics[width=0.12\textwidth]{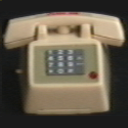}\label{figObj60}}
\subfigure[\vspace{-9pt}Reconstruction performance of \texttt{Object}
\texttt{\#16}]{\includegraphics[width=0.48\textwidth]{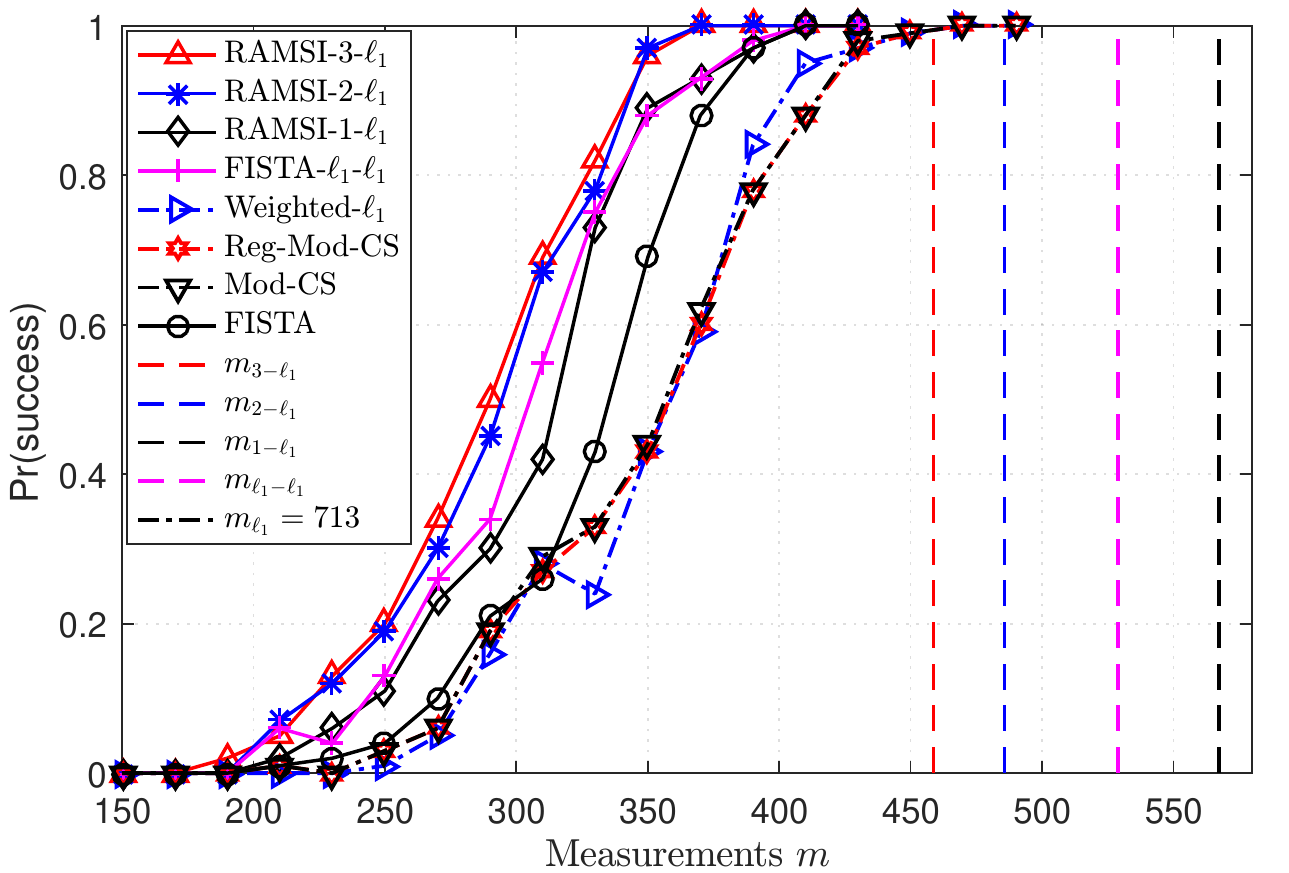}\label{figPerObj16}}
\subfigure[Reconstruction performance of \texttt{Object} \texttt{\#60}]{\includegraphics[width=0.48\textwidth]{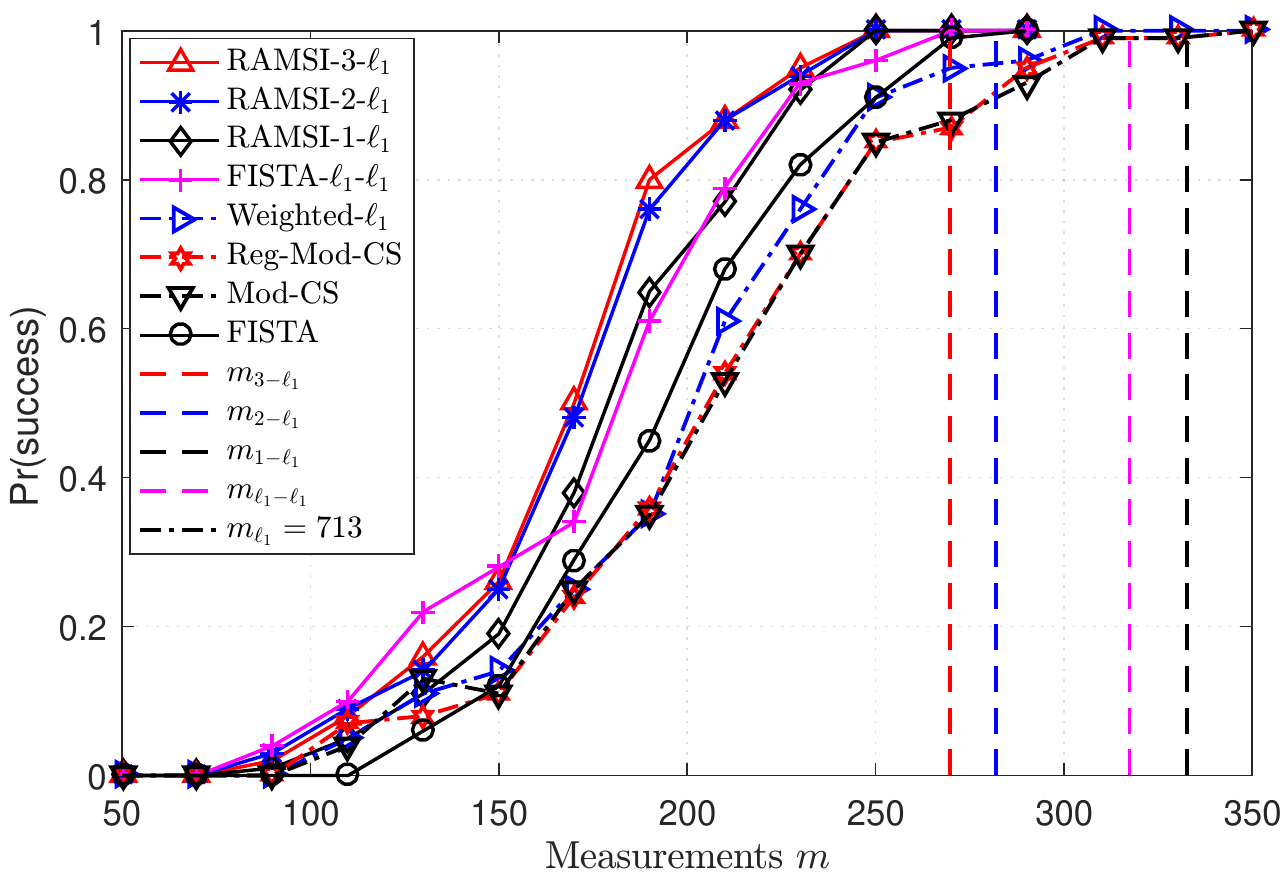}\label{figPerObj60}}
\caption{Successful probabilities and measurement bounds vs. number of measurements $m$ for \texttt{Object}
\texttt{\#16} and \texttt{Object} \texttt{\#60} in \texttt{COIL-100} \cite{Nene96}.}\label{figFeatHis}
\end{figure}

We now evaluate the RAMSI configurations and the derived bounds on sparse sources extracted from the multiview image database \cite{Nene96}. Fig. \ref{figFeatHis} presents the performance of RAMSI, with one, two and three side information signals, against that of FISTA \cite{Beck09}, FISTA-$\ell_{1}$-$\ell_{1}$ \cite{MotaGLOBALSIP14}, \textcolor{black}{Weighted-$\ell_{1}$ \cite{KhajehnejadXAH11}, Reg-Mod-CS \cite{NVaswaniTSP12}, and Mod-CS \cite{NVaswaniTSP10}}. \textcolor{black}{The support estimate that is considered as prior information by the Mod-CS, Reg-Mod-CS, and Weighted-$\ell_{1}$ algorithms is given by the nearest view, which corresponds to the highest quality side information signal}. The reconstruction performance is expressed in terms of the probability of success---\textcolor{black}{the number of times divided by the total number of 100 Monte Carlo iterations}---versus the number of measurements. As in Section \ref{generatedPer}, we set $\epsilon=10^{-5}$ and $\lambda=10^{-5}$. Fig. \ref{figPerObj16} and Fig. \ref{figPerObj60} show clearly that RAMSI significantly improves the reconstruction accuracy for the signals extracted from \texttt{Object}
\texttt{\#16} [Fig. \ref{figObj16}] and \texttt{Object} \texttt{\#60} [Fig.
\ref{figObj60}]. \textcolor{black}{Furthermore, the performance of RAMSI systematically increases with the number of side information signals used to aid the reconstruction. It is worth noting that the improvement brought by RAMSI-3-$\ell_{1}$ over RAMSI-2-$\ell_{1}$ is not remarkable as the third side information signal is extracted from the furthest neighboring view. It is also important to mention that, in this practical scenario, the Weighted-$\ell_{1}$ \cite{KhajehnejadXAH11}, Reg-Mod-CS \cite{NVaswaniTSP12}, and Mod-CS \cite{NVaswaniTSP10} methods have worse performance than FISTA because of the poor support estimate}.

\section{Conclusion}
\label{conclusion}
This paper established theoretical measurement bounds for the problem of signal recovery under a weighted $n\text{-}\ell_1$ minimization framework and proposed the RAMSI algorithm. The proposed algorithm incorporates multiple side information signals in the problem of sparse signal recovery and iteratively weights the various signals so as to optimize the reconstruction performance. The bounds confirm the advantage of RAMSI in utilizing multiple side information signals to significantly reduce the number of measurements and to deal with variations in the quality of the side information. We experimentally assessed the established bounds and the performance of RAMSI against state-of-the-art methods using both synthetic and real-life sparse signals. The results showed that the measurement bounds are sharper than existing bounds and that RAMSI outperforms the conventional CS method, the recent $\ell_{1}$-$\ell_{1}$ minimization method as well as the Modified-CS, regularized Modified-CS, and weighted $\ell_1$ minimization methods. Moreover, RAMSI can efficiently incorporate multiple side information signals, with its performance typically increasing with the number of the involved side information signals.

\appendices
\section{Proof of Proposition \ref{propMSI}}\label{appendixProximal}
Replacing~$g(\bx)=\lambda \sum_{j=0}^{J}\|\mathbf{W}_{j}(\bx-\bz_{j})\|_{1}$ in the definition of the proximal operator in~\eqref{l1-proximalOperator}, we obtain
\begin{equation}\label{n-l1-proximalOperatorCompute}
\Gamma_{\frac{1}{L}g}(\bx)=\arg\hspace{-1pt}\min_{\bv \in\mathbb{R}^{n}}\Big\{\frac{\lambda}{L} \sum\limits_{j=0}^{J}\|\mathbf{W}_{j}(\bv-\bz_{j})\|_{1}+ \frac{1}{2}\|\bv-\bx\|^{2}_{2}\Big\}.
\end{equation}
Both terms in \eqref{n-l1-proximalOperatorCompute} are separable in $\bv$ and thus, we can minimize each element $v_{i}$ of $\bv$ individually as
\begin{equation}\label{n-l1-proximalOperatorElement}
\Gamma_{\hspace{-2pt}\frac{1}{L}g}(x_{i})\hspace{-1pt} =\hspace{-1pt} \argmin_{v_{i} \in\mathbb{R}}\hspace{-1pt}\Big\{\hspace{-1pt}h(v_{i})\hspace{-1pt}=\hspace{-1pt}\frac{\lambda}{L}\hspace{-1pt} \sum\limits_{j=0}^{J}w_{ji}|v_{i}\hspace{-1pt}-\hspace{-1pt}z_{ji}|\hspace{-1pt} + \hspace{-2pt} \frac{1}{2}(v_{i}\hspace{-1pt}-\hspace{-1pt}x_{i})^{2}\hspace{-1pt}\Big\}.
\end{equation}

To solve~\eqref{n-l1-proximalOperatorElement} we need to derive $\frac{\partial h(v_{i})}{\partial v_{i}}$. Recall from Assumption \ref{asPartical} that, without loss of generality, we assume that $-\infty=z_{(-1)i} \leq z_{0i}\leq z_{1i} \leq\dots \leq z_{Ji} \leq z_{(J+1)i}=\infty$. When $v_{i}\in(z_{li},z_{(l+1)i})$ with $l\in\{-1,\dots,J\}$, $h(v_{i})$ is differentiable and $\frac{\partial h(v_{i})}{\partial v_{i}}$ is calculated as
\begin{equation}\label{n-l1-proximalOperatorElementDerivative}
\frac{\partial h(v_{i})}{\partial v_{i}}= \frac{\lambda}{L} \sum\limits_{j=0}^{J}w_{ji}\mathrm{sign}(v_{i}-z_{ji})+ (v_{i}-x_{i}),
\end{equation}
where $\mathrm{sign}(.)$ is the sign function. Using the definition of the $\mathfrak{b}(.)$ function in \eqref{eq:bfunction}, we have that $\mathrm{sign}(v_{i}-z_{ji})=(-1)^{\mathfrak{b}(l<j)}$ and thus, we can rewrite \eqref{n-l1-proximalOperatorElementDerivative} as
\begin{equation}\label{n-l1-proximalOperatorElementBoolean}
\frac{\partial h(v_{i})}{\partial v_{i}}=\frac{\lambda}{L}\sum\limits_{j=0}^{J}w_{ji}(-1)^{\mathfrak{b}(l<j)}+ (v_{i}-x_{i}).
\end{equation}
By setting $\frac{\partial h(v_{i})}{\partial v_{i}}=0$, we obtain:
\begin{equation}\label{partialZero}
v_{i} = x_{i}-\frac{\lambda}{L} \sum_{j=0}^{J}w_{ji}(-1)^{\mathfrak{b}(l<j)}.
\end{equation}
Since $v_{i}\in(z_{li},z_{(l+1)i})$, via \eqref{partialZero} we have that \begin{equation}\label{n-l1-proximalX}
z_{li}+\frac{\lambda}{L}\sum\limits_{j=0}^{J}w_{ji}(-1)^{\mathfrak{b}(l<j)}\hspace{-1pt}<\hspace{-1pt}x_{i}\hspace{-1pt}<\hspace{-1pt}z_{(l+1)i}\hspace{-1pt}+\hspace{-1pt}\frac{\lambda}{L}\hspace{-1pt} \sum\limits_{j=0}^{J}w_{ji}(\hspace{-1pt}-1\hspace{-1pt})^{\mathfrak{b}(l<j)}.
\end{equation}

In the following Lemma \ref{lemmaMin}, we prove that, in the remaining range value of $x_i$, namely, in the case that
\begin{equation}\label{n-l1-proximalZ}
z_{li}\hspace{-1pt}+\hspace{-1pt}\frac{\lambda}{L} \hspace{-1pt}\sum\limits_{j=0}^{J}w_{ji}(\hspace{-1pt}-1\hspace{-1pt})^{\mathfrak{b}(l-1<j)}\hspace{-1pt}\leq \hspace{-1pt}x_{i}\hspace{-1pt}\leq \hspace{-1pt} z_{li}\hspace{-1pt}+\hspace{-1pt}\frac{\lambda}{L} \hspace{-1pt}\sum\limits_{j=0}^{J}w_{ji}(\hspace{-1pt}-1\hspace{-1pt})^{\mathfrak{b}(l<j)},
\end{equation}
the minimum of $h(v_{i})$ is obtained when $v_{i}=z_{li}$.
\begin{lemma}
\label{lemmaMin}
Given that $x_{i}$ is bounded as in \eqref{n-l1-proximalZ}, the function $h(v_{i})$ defined in \eqref{n-l1-proximalOperatorElement} is minimized when $v_{i}=z_{li}$.
\end{lemma}
\begin{proof}
Let us express $h(v_{i})$ as
\vspace{-8pt}
\begin{equation}
\label{n-l1-proximalOperatorElementLemma}
h(\hspace{-1pt}v_{i}\hspace{-1pt})\hspace{-1pt} = \hspace{-1pt}\frac{\lambda}{L}\hspace{-1pt} \sum_{j=0}^{J}w_{ji}|(\hspace{-1pt}v_{i}\hspace{-1pt}-\hspace{-1pt}z_{li}\hspace{-1pt})\hspace{-1pt}-\hspace{-1pt}(\hspace{-1pt}z_{ji}\hspace{-1pt}-\hspace{-1pt}z_{li}\hspace{-1pt})| \hspace{-1pt}+ \hspace{-1pt}\frac{1}{2}[(\hspace{-1pt}v_{i}\hspace{-1pt}-\hspace{-1pt}z_{li}\hspace{-1pt})\hspace{-1pt}-\hspace{-1pt}(\hspace{-1pt}x_{i}\hspace{-1pt}-\hspace{-1pt}z_{li}\hspace{-1pt})]^{2}.
\end{equation}
Applying the inequality $|a-b|\geq |a|-|b|$, with $a,b\in \mathbb{R}$, on the first summation term and expanding the second term, we obtain:
\begin{align}
\label{n-l1-proximalOperatorElementLemmaInequality}
h(v_{i})&\geq\frac{\lambda}{L}\sum\limits_{j=0}^{J}w_{ji}|v_{i}-z_{li}|-\frac{\lambda}{L} \sum\limits_{j=0}^{J}w_{ji}|z_{ji}-z_{li}| \nonumber\\
&\quad+\frac{1}{2}(v_{i}-z_{li})^{2}-(v_{i}-z_{li})\cdot(x_{i}-z_{li})+\frac{1}{2}(x_{i}-z_{li})^{2}.
\end{align}
Using the basic inequality $-(v_{i}-z_{li})\cdot(x_{i}-z_{li})\geq -|v_{i}-z_{li}|\cdot|x_{i}-z_{li}|$, we can express \eqref{n-l1-proximalOperatorElementLemmaInequality} as
\begin{equation}
\label{n-l1-proximalOperatorElementLemmaInequalityFurther}
\begin{split}
h(v_{i})\hspace{-1pt} \geq &|v_{i}\hspace{-1pt}-\hspace{-1pt}z_{li}|\frac{\lambda}{L}\hspace{-1pt}\sum\limits_{j=0}^{J}\hspace{-1pt}w_{ji} \hspace{-1pt}-\hspace{-1pt}|v_{i}\hspace{-1pt}-\hspace{-1pt}z_{li}|\cdot|x_{i}\hspace{-1pt}-\hspace{-1pt}z_{li}|\hspace{-1pt}+ \hspace{-1pt} \frac{1}{2}(v_{i}\hspace{-1pt}-\hspace{-1pt}z_{li})^{2}\\
&\qquad\qquad\qquad-\frac{\lambda}{L} \sum\limits_{j=0}^{J}w_{ji}|z_{ji}\hspace{-1pt}-\hspace{-1pt}z_{li}| \hspace{-1pt}+\hspace{-1pt}\frac{1}{2}(x_{i}\hspace{-1pt}-\hspace{-1pt}z_{li})^{2}.
\end{split}
\end{equation}
At this point, via~\eqref{n-l1-proximalZ} we obtain
\begin{align}\label{n-l1-proximalZrewriteMore}
&-\frac{\lambda}{L}\sum\limits_{j=0}^{J}w_{ji}\leq x_{i}-z_{li}\leq \frac{\lambda}{L}\sum\limits_{j=0}^{J}w_{ji}\nonumber\\
\Rightarrow&|x_{i}-z_{li}|\leq \frac{\lambda}{L}\sum\limits_{j=0}^{J}w_{ji}.
\end{align}
We now observe that the part including $v_{i}$ in the right hand side of \eqref{n-l1-proximalOperatorElementLemmaInequalityFurther} can be written as
\begin{equation}\label{leftInequality}
|v_{i}-z_{li}|\Big(\frac{\lambda}{L}\sum\limits_{j=0}^{J}w_{ji} - |x_{i}-z_{li}|\Big) +  \frac{1}{2}(v_{i} - z_{li})^{2}.
\end{equation}
 Taking into account \eqref{n-l1-proximalZrewriteMore}, the expression in \eqref{leftInequality} and, in turn $h(v_{i})$, are minimized when $v_{i} = z_{li}$.
\end{proof}

In summary, via \eqref{partialZero} and Lemma \ref{lemmaMin}, we obtain
\begin{equation}\label{n-l1-proximalOperatorElementCompute}
\Gamma_{\frac{1}{L}g}(x_{i}) = \left\{
\begin{array}{l}
x_{i} - \frac{\lambda}{L}\sum\limits_{j=0}^{J}w_{ji}(-1)^{\mathfrak{b}(l<j)},\mathrm{~~~if}\:\:\eqref{n-l1-proximalX}\\[-3pt]
z_{li},~~~~~~~~~~~~~~~~~~~~~~~~~~~~~~\mathrm{if}\:\:\eqref{n-l1-proximalZ}.
\end{array}
\right.
\end{equation}

\textcolor{black}{
\section{Proof of Theorem~\ref{RAMSIBound} and Corollary~\ref{corBoundRelation}}}
\label{proof:ProofTheorem4}

\textcolor{black}{We begin by stating some important results that help us proving Theorem~\ref{RAMSIBound} and Corollary~\ref{corBoundRelation}.}

Recall that the probability density of the normal distribution $\mathcal{N}(0,1)$
with zero-mean and unit variance $\psi(x)$ is given by
\begin{equation}\label{densityDistribution}
 \psi(x):=\frac{1}{\sqrt{2\pi}}e^{-x^{2}/2}.
\end{equation}

Our formulations consider the following inequality, which is stated in~\cite{MotaARXIV14},
that is
 \begin{equation}
\label{logInequality}
    \frac{(1-x^{-1})}{\sqrt{\pi\log(x)}}\leq\frac{1}{\sqrt{2\pi}}\leq\frac{2}{5},
\end{equation}
for all $x>1$. Moreover, adhering to the formulations in~\cite{MotaARXIV14},
we use the following expressions in our derivations:
\begin{subequations}
\label{spetralFormulation}
  \begin{align}
  &\mathcal{A}(x):=\frac{1}{\sqrt{2\pi}}\int_{x}^{\infty}(v-x)^{2}e^{-v^{2}/2}\mathrm{d}v\label{positiveSpectral}\\
  &\mathcal{B}(x):=\frac{1}{\sqrt{2\pi}}\int_{-\infty}^{x}(v-x)^{2}e^{-v^{2}/2}\mathrm{d}v,
\label{negativeSpectral}
\end{align}
\end{subequations}
for which we have that
\begin{equation}\label{zeroABQ}
  \mathcal{A}(0)=\mathcal{B}(0)=1/2.
\end{equation}
When $x\neq 0$, the following inequalities have been derived in~\cite{MotaARXIV14}:
\begin{subequations}
\label{spetralInequality}
  \begin{align}
  &    \mathcal{A}(x)\leq\left\{
  \begin{array}{l}
    \psi(x)/x,~~~~~~x>0 \\
    x^{2}+1  ,~~~~~~~x<0
  \end{array}
\right.\label{positiveA}\\
  &\mathcal{B}(x)\leq\left\{
  \begin{array}{l}
    -\psi(x)/x,~~~~x<0 \\
    x^{2}+1  ,~~~~~~~x>0.
  \end{array}
\right. \label{negativeB}
  \end{align}
\end{subequations}

\begin{lemma}
\label{lemmaPsi}
Given $x \in (0,1]$ and $\tau >0$ for $\psi(x)$ given in~\eqref{densityDistribution},
we have
\begin{equation}
\label{psiInequality}
    \frac{\psi(\tau x)}{\tau x}\leq \frac{1}{\sqrt{2\pi}}\frac{1-x^{2}}{
\tau x}+x\frac{\psi(\tau)}{\tau }.
\end{equation}
\end{lemma}
\begin{proof}
Using~\eqref{densityDistribution} the left hand side of \eqref{psiInequality}
becomes
\begin{equation}
\label{psiInequalityDeploy}
    \frac{\psi(\tau x)}{\tau x}=\frac{1}{\sqrt{2\pi}}\frac{e^{-\tau^{2}x^{2}/2}}{\tau
x}.
\end{equation}
Applying Bernoulli's inequality on $e^{-\tau^{2}x^{2}/2}$ leads to
\begin{equation}\label{bernoulli}
    e^{-\tau^{2}x^{2}/2}=\Big(1+(e^{-\tau^{2}/2}-1)\Big)^{x^{2}}\leq 1+x^{2}(e^{-\tau^{2}/2}-1),
\end{equation}
where $0<x\leq 1$ and $(e^{-\tau^{2}/2}-1)> -1$ given that $\tau >0$. Combining~\eqref{psiInequalityDeploy} and~\eqref{bernoulli} leads to the proof.
\end{proof}

\begin{proof}[Proof of Theorem~\ref{RAMSIBound}]
\textcolor{black}{We derive the bound based on Proposition \ref{propUpper} by firstly computing the subdifferential $\partial g(\bx)$ and then the distance between the standard normal vector $\bg$ to $\partial g(\bx)$. Under Assumptions
\ref{asIndex}, \ref{asIndexZero}, \ref{asPartical}} and the weights in \eqref{n-l1-weights},
the elements of the vectors $\bu$ in the subdifferential $\partial g(\bx)$ are expressed as
\begin{equation}\label{csSubdifferential}
  \begin{array}{l}
    u_{i}= a_{i},~~~~~~~~~~~~~~~~~~~~~~~~~~~~~~~~~~~i\in\{1,\dots,p\}, \\
    u_{i}\in [b_{i}- c_{i},b_{i}+c_{i}],~~~~~~~~~~~~~~~~~~~i\in\{p+1,\dots,q\},\\
    u_{i} \in \big[-\sum\limits_{j=0}^{J} w_{ji}, \sum\limits_{j=0}^{J}
w_{ji}\big] = [-1,1],~~i \in\{q + 1,\dots,n\}.\end{array}
\end{equation}
where
\begin{subequations}
\label{abEquality}
\begin{align}
&a_{i}=\sum\limits_{j=0}^{J} w_{ji}( -1)^{\mathfrak{b}(l_{i}<j)} \label{a}\\
&b_{i}= \sum_{j \notin\{l_{i},\dots, l_{i} + d_{i}-1\}} w_{ji}(  -1 )^{\mathfrak{b}(l_{i}<j)}\label{b}\\
&c_{i}= \sum_{j=l_{i}}^{l_{i}+d_{i}-1}  w_{ji}=d_{i}\frac{\eta_{i}}{\epsilon}.\label{c}
\end{align}
\end{subequations}

Using \eqref{euclideanDistance} and \eqref{csSubdifferential}, we can compute the distance from the standard normal vector $\bg$ to the subdifferential $\partial g(\bx)$ as
\begin{align}\label{distanceToSubdifferentialCS}
\mathrm{dist}^{2} &(\bg,\tau  \cdot \partial g(\bx))=\nonumber\\
&\sum\limits_{i=1}^{p}(\g_{i}-\tau a_{i})^{2} \nonumber\\
&+\sum\limits_{i=p+1}^{q}\hspace{-7pt}\Big( \mathcal{P}^{2} (\g_{i} -
      \tau(b_{i} + c_{i}))
      + \mathcal{P}^{2} ( -\g_{i} + \tau(b_{i} - c_{i})) \Big)
     \nonumber\\
     &+ \sum\limits_{i=q+1}^{n} \mathcal{P}^{2}(|\g_{i}|-\tau),
\end{align}
where $\mathcal{P}(a) := \max\{a,0\}$ returns the maximum value between $a \in \mathbb{R}$ and $0$.
Taking the expectation of \eqref{distanceToSubdifferentialCS} with respect to $\bg$ delivers
\begin{align}\label{distanceToSubdifferentialExpectationCS}
\mathbb{E}_{\bg}[\mathrm{dist}^{2}&(\bg,\tau \cdot \partial g(\bx))]= \nonumber\\
&p+\tau^{2}\sum\limits_{i=1}^{p} a_{i}^{2}\nonumber\\
&+\frac{1}{\sqrt{2\pi}}\sum\limits_{i=p+1}^{q} \int_{\tau(b_{i} +c_{i})}^{\infty} (v-\tau (b_{i}+ c_{i}))^{2}\mathrm{e}^{-v^{2} /2}\mathrm{d}v\nonumber\\
&+\frac{1}{\sqrt{2\pi}}\sum\limits_{i=p+1}^{q} \int^{\tau(b_{i}-c_{i})
}_{-\infty} (v-\tau (b_{i} - c_{i}))^{2}\mathrm{e}^{-v^{2} /2}\mathrm{d}v\nonumber\\
&+\sqrt{\frac{2}{\pi}} \sum\limits_{i=q+1}^{n}\int_{\tau}^{\infty}(v-\tau)^{2}\mathrm{e}^{-v^{2} /2}\mathrm{d}v.
\end{align}
Replacing the expressions in \eqref{positiveSpectral}, \eqref{negativeSpectral}
 in \eqref{distanceToSubdifferentialExpectationCS}
gives
\begin{align}\label{distanceToSubdifferentialExpectationCSRepresent}
\mathbb{E}_{\bg}[\mathrm{dist}^{2} (\bg,&\tau\cdot\partial g(\bx))]= \nonumber\\
&p+\tau^{2} \sum\limits_{i=1}^{p} a_{i}^{2}+\sum\limits_{i=p+1}^{q} \mathcal{A}\Big(\tau(b_{i} +c_{i})\Big)\nonumber\\
&+ \sum\limits_{i=p+1}^{q} \mathcal{B}\Big(\tau(b_{i} -c_{i})\Big)+2 \sum\limits_{i=q+1}^{n} \mathcal{A}(\tau).
\end{align}

At this point, it is worth emphasizing the advantage of the adaptive weights [see ~\eqref{n-l1-weights}] in the proposed method, the values of which depend on the correlation of the side information with the target signal. Focusing on a given index $i\in\{p+1,\dots ,q\}$,
let us observe the weight contribution in the expressions of $b_{i}$
and $c_{i}$, defined in \eqref{b} and \eqref{c}, respectively. The weights $w_{ji}$ in $c_{i}$ are considerably higher than those in $b_{i}$, as in the former case the side information
signals $z_{ji}$ are equal to the source $x_{i}$. Moreover, we recall that a
small positive parameter $\epsilon$ is introduced in the denominator of the weights so as to avoid division by zero when $x_{i}=z_{ji}$. The $\epsilon$ parameter can ensure that $c_{i}$ in \eqref{c}
is always greater than $|b_{i}|$ in \eqref{b}. Hence, we always have $b_{i} + c_{i} > 0$ and $b_{i}-c_{i}<0$.
With these observations we conclude that the arguments of~$\mathcal{A}(\cdot)$ and ~$\mathcal{B}(\cdot)$ in~\eqref{distanceToSubdifferentialExpectationCSRepresent} are respectively positive and negative. Applying inequality~\eqref{positiveA} for $x\hspace{0pt}>\hspace{0pt}0$ on the expression of $\mathcal{A}(\cdot)$ as well as inequality \eqref{negativeB} for $x<0$ on the expression of $\mathcal{B}(\cdot)$, we obtain the following bound for the $U_g$ quantity [which is defined in~\eqref{upperBoundCompute}]:
\begin{align}\label{upperBoundComputeCS}
&U_{g_{n\text{-}\ell_{1}}}\leq\min_{\tau\geq 0}\Big\{p+\tau^{2}\sum\limits_{i=1}^{p}a_{i}^{2}\nonumber\\
&+\sum\limits_{i=p+1}^{q}\Big(\frac{\psi(\tau(b_{i}+c_{i}))}{\tau (b_{i}+c_{i})}+\frac{\psi(\tau
(c_{i}-b_{i}))}{\tau (c_{i}-b_{i})}\Big)+2\sum\limits_{i=q+1}^{n}\frac{\psi(\tau)}{\tau}\Big\},
\end{align}
where $\psi(\cdot)$ is the zero-mean, unit-variance normal distribution defined in~\eqref{densityDistribution}.
Applying \eqref{psiInequality} on the second sum in \eqref{upperBoundComputeCS} gives
\begin{align}\label{upperBoundComputeCSApplyPsi}
\negmedspace U_{g_{n\text{-}\ell_{1}}}&\leq\min_{\tau\geq 0}\Big\{p+\tau^{2}\sum\limits_{i=1}^{p}a_{i}^{2}\nonumber\\
&+\sum\limits_{i=p+1}^{q}\Big(\frac{1-(b_{i}+c_{i})^{2}}{\tau\sqrt{2\pi}
(b_{i}+c_{i})}+ \frac{1-(c_{i}-b_{i})^{2}}{\tau\sqrt{2\pi}(c_{i}-b_{i})}+c_{i}\frac{2\psi(\tau)}{\tau}\Big)\nonumber\\
&+2(n-q)\frac{\psi(\tau)}{\tau}\Big\}.
\end{align}
\begin{align}\label{upperBoundComputeCSApplyPsiMore}
\Rightarrow U_{g_{n\text{-}\ell_{1}}}\leq &\min_{\tau\geq 0}\Big\{p+\tau^{2}\sum\limits_{i=1}^{p}a_{i}^{2}\nonumber\\
&+\sum\limits_{i=p+1}^{q}\frac{1}{\sqrt{2\pi}}\frac{2c_{i}}{\tau}\Big(\frac{1}{c_{i}^{2}-b_{i}^{2}}-1\Big)\nonumber\\
&+\Big((n-q)+\sum\limits_{i=p+1}^{q}c_{i}\Big)2\frac{\psi(\tau)}{\tau}\Big\}.
\end{align}
From the definitions of $b_{i}$ and $c_{i}$ in \eqref{b} and \eqref{c}, respectively, we have
\begin{equation}\label{bcInequality}
b_{i}\leq \sum_{j \notin\{l_{i},\dots,l_{i}+d_{i}-1\}} w_{ji}=1-c_{i},
\end{equation}
where we used the constraint $\sum_{j=0}^{J}w_{ji}=1$. Hence, the second summation term in \eqref{upperBoundComputeCSApplyPsiMore}
is bounded as
\begin{align}
\sum\limits_{i=p+1}^{q}\frac{1}{\sqrt{2\pi}}&\frac{2c_{i}}{\tau}\Big(\frac{1}{c_{i}^{2}-b_{i}^{2}}-1\Big)\nonumber\\
&\leq\sum\limits_{i=p+1}^{q}\hspace{-4pt}\frac{4}{\sqrt{2\pi} \tau}\frac{c_{i}}{2c_{i}-1}(1-c_{i})\label{cbSum}\\
&\leq\frac{4\cdot\min \{c_{i}\}}{\sqrt{2\pi}\tau(2\cdot\min\{ c_{i}\}-1)}\sum\limits_{i=p+1}^{q}(1-c_{i}),\label{cbSum22}
\end{align}
where \eqref{cbSum} is obtained by using that $b_i\leq1-c_{i}$ and \eqref{cbSum22} holds from the fact that the term $\frac{c_{i}}{2c_{i}-1}$ is maximized when $c_i$ is minimized (recall that $c_i>0$).
For simplicity, let us denote
\begin{subequations}\label{denoteQuantity}
\begin{align}
\bar{a}_{n\text{-}\ell_{1}}&=\sum\limits_{i=1}^{p} a_{i}^{2}\label{hbarQuantity}\\
\kappa_{n\text{-}\ell_{1}} &= \frac{4\cdot\min \{c_{i}\}}{\sqrt{2\pi}\tau(2\cdot\min\{
c_{i}\}-1)}\label{kappaQuantity}\\
\bar{s}_{n\text{-}\ell_{1}}&=q- \sum\limits_{i=p+1}^{q} c_{i}=p+ \sum\limits_{i=p+1}^{q} (1-c_{i})
\label{sparseQuantity}.
\end{align}
\end{subequations}
Substituting the quantities of \eqref{hbarQuantity}, \eqref{kappaQuantity},
\eqref{sparseQuantity}, and using inequality \eqref{cbSum22} in \eqref{upperBoundComputeCSApplyPsiMore}
gives
\begin{align}\label{upperBoundComputeCSApplyPsiMoreSimplicity}
 U_{g_{n\text{-}\ell_{1}}}\leq\min_{\tau\geq 0}\Big\{\bar{a}_{n\text{-}\ell_{1}}&\tau^{2}+(n-\bar{s}_{n\text{-}\ell_{1}})\frac{2\psi(\tau)}{\tau}\nonumber\\
& +p+\kappa_{n\text{-}\ell_{1}} (\bar{s}_{n\text{-}\ell_{1}}-p)\Big\},
\end{align}
which using the definition of $\psi(\cdot)$ can be written as
\begin{align}\label{upperBoundComputeCSApplyPsiMoreSimplicity}
U_{g_{n\text{-}\ell_{1}}}\leq\min_{\tau\geq 0}\Big\{\bar{a}_{n\text{-}\ell_{1}}&\tau^{2}+(n-\bar{s}_{n\text{-}\ell_{1}})\frac{2}{\sqrt{2\pi}}\frac{e^{-\frac{\tau^{2}}{2}}}{\tau}\nonumber\\
&+\bar{s}_{n\text{-}\ell_{1}}+(\kappa_{n\text{-}\ell_{1}}-1)(\bar{s}_{n\text{-}\ell_{1}}-p)\Big\}.
\end{align}
To derive a bound as a function of the source signal $\bx$,
we need to select a parameter $\tau>0$. Setting $\tau=\sqrt{2\log(n/\bar{s}_{n\text{-}\ell_{1}})}$ yields
\begin{equation}\label{upperBoundComputeCSApplyPsiMoreSimplicityMore}
U_{g_{n\text{-}\ell_{1}}}\hspace{-3pt}\leq
2\bar{a}_{n\text{-}\ell_{1}}\log\frac{n}{\bar{s}_{n\text{-}\ell_{1}}}\hspace{-1pt}+\hspace{-1pt}
\frac{\bar{s}_{n\text{-}\ell_{1}}(1\hspace{-1pt}-\hspace{-1pt}\bar{s}_{n\text{-}\ell_{1}}/n)}{\sqrt{\pi
\log(n/\bar{s}_{n\text{-}\ell_{1}})}}\hspace{-1pt}+\hspace{-1pt}\bar{s}_{n\text{-}\ell_{1}}\hspace{-1pt}+\hspace{-1pt}\delta_{n\text{-}\ell_{1}},
\end{equation}
where
\begin{equation}\label{deltaNL1}
\delta_{n\text{-}\ell_{1}}= (\kappa_{n\text{-}\ell_{1}}-1)(\bar{s}_{n\text{-}\ell_{1}}-p),
\end{equation}
and where we have replaced the selected value of $\tau$ in~\eqref{kappaQuantity}, thereby obtaining the $\kappa_{n\text{-}\ell_{1}}$ definition reported in~\eqref{eq:KappaFinalDefinition}.

Applying inequality \eqref{logInequality} on the second term of the right hand side of \eqref{upperBoundComputeCSApplyPsiMoreSimplicityMore}
gives
\begin{equation}\label{upperBoundComputeCSApplyPsiMoreSimplicityApply}
U_{g_{n\text{-}\ell_{1}}}\leq2\bar{a}_{n\text{-}\ell_{1}} \log\frac{n}{\bar{s}_{n\text{-}\ell_{1}}}+\frac{7}{5}
\bar{s}_{n\text{-}\ell_{1}}+\delta_{n\text{-}\ell_{1}}.
\end{equation}
Bearing in mind that $m_{n\text{-}\ell_{1}}\geq U_{g_{n\text{-}\ell_{1}}}\hspace{0pt}+1$ and by combining~\eqref{upperBoundComputeCSApplyPsiMoreSimplicityApply} with~\eqref{denoteQuantity},~\eqref{abEquality}, and~\eqref{deltaNL1} leads to the proof.
\end{proof}

\vspace{8pt}

\begin{proof}[Proof of Corollary~\ref{corBoundRelation}]
\textcolor{black}{
We start with Relation \eqref{cs-l1Relation}. Under the conditions that $\mathbf{W}_{0}=\mathbf{I}_{n}$ and $\mathbf{W}_{j}=\mathbf{0}$ for $j\in\{1,\dots,J\}$, we have that $\bar{s}_{n\text{-}\ell_{1}}=p=s_{0}$ and $\bar{a}_{n\text{-}\ell_{1}}=p=s_{0}$, where we used the definitions in \eqref{hbarQuantity}, \eqref{sparseQuantity}.
Consequently, from \eqref{deltaNL1}, $\delta_{n\text{-}\ell_{1}}=0$. Replacing these values of $\bar{a}_{n\text{-}\ell_{1}}$, $\bar{s}_{n\text{-}\ell_{1}}$, $\delta_{n\text{-}\ell_{1}}$ in our bound, defined in \eqref{upperBoundComputeCSApplyAEtaSetMeasurement}, leads to the $\ell_{1}$ minimization bound in \eqref{l1 bound}.}

To reach Relation \eqref{cs-l1-l1Relation}, given that $\mathbf{W}_{0}=\mathbf{W}_{1}=\frac{1}{2}\mathbf{I}_n$
and $\mathbf{W}_{j}=\mathbf{0}$ for $j\in\{2,\dots,J\}$,
let us first denote two subsets, $\mathds{I}_{1}$ and $\mathds{I}_{2}$, as
\begin{subequations}\label{denoteL1-L1-bc}
\begin{align}
\negmedspace\negthickspace\mathds{I}_{1}&:=\Big\{i\in\{p+1,\dots,q\}:b_{i}+c_{i}=1,b_{i}-c_{i}=0\Big\}\label{I1Quantity}\\
\negmedspace\negthickspace\mathds{I}_{2}&:=\Big\{i\in\{p+1,\dots,q\}:b_{i}+c_{i}=0,b_{i}-c_{i}=-1\Big\}\label{I2Quantity}.
\end{align}
\end{subequations}
Via the definitions of $b_{i}$ and $c_{i}$---see~\eqref{b}~and~\eqref{c}, respectively---and since $i\in\{p+1,\dots,q\}$, we observe that $i\in\mathds{I}_{1}$ or $i\in\mathds{I}_{2}$.

Replacing \eqref{I1Quantity} and \eqref{I2Quantity} in \eqref{distanceToSubdifferentialExpectationCSRepresent}
leads to
\begin{align}\label{distanceToSubdifferentialExpectationCSRepresentRelation}
\mathbb{E}_{\bg}[\mathrm{dist}^{2}&(\bg,\tau\hspace{-1pt}\cdot\hspace{-1pt}\partial g(\bx))] \nonumber\\
&= p+\tau^{2} \sum\limits_{i=1}^{p}a_{i}^{2}+\sum\limits_{i\in\mathds{I}_{1}}\mathcal{A}(\tau)
+\sum\limits_{i\in\mathds{I}_{1}}\mathcal{B}(0)+\sum\limits_{i\in\mathds{I}_{2}}\mathcal{A}(0)\nonumber\\
&\qquad\qquad\qquad\qquad+\sum\limits_{i\in\mathds{I}_{2}}\mathcal{B}(-\tau)+2\sum\limits_{i=q+1}^{n}\mathcal{A}(\tau).
\end{align}
By combining \eqref{zeroABQ}, \eqref{positiveA}, and \eqref{negativeB} with \eqref{distanceToSubdifferentialExpectationCSRepresentRelation}, we obtain the following bound
for the $U_g$ quantity [defined in~\eqref{upperBoundCompute}]:
\begin{align}\label{distanceToSubdifferentialExpectationCSRepresentRelationApply}
U_{g_{\ell_{1}\text{-}\ell_{1}}}\leq \min_{\tau\geq 0}\Big\{p+\tau^{2}&\sum\limits_{i=1}^{p}a_{i}^{2}
+\sum\limits_{i=p+1}^{q}\frac{1}{2}\nonumber\\
&+\sum\limits_{i=p+1}^{q}\frac{\psi(\tau)}{\tau}+2\sum\limits_{i=q+1}^{n}\frac{\psi(\tau)}{\tau}\Big\},
\end{align}
which can further be elaborated to
\begin{align}\label{distanceToSubdifferentialExpectationCSRepresentRelationApply}
U_{g_{\ell_{1}\text{-}\ell_{1}}}&\leq\min_{\tau\geq 0}\Big\{\tau^{2}\sum\limits_{i=1}^{p}a_{i}^{2}+\frac{1}{2}(p+q)+(2n-(p+q))\frac{e^{\frac{-\tau^{2}}{2}}}{\sqrt{2\pi}\tau}\Big\}.
\end{align}
In the $\ell_1\text{-}\ell_1$ minimization case, there is a single side information signal and no weights, that is, $d_{i}=1$ and $J=1$ in \eqref{csSparseEquality} under Assumptions~\ref{asIndex} and~\ref{asIndexZero}; thus, we have $p\hspace{0pt}+\hspace{0pt}q
\hspace{-1pt}= \hspace{-1pt}s_{0}\hspace{0pt}+\hspace{0pt}s_{1}$. Combining
this result with \eqref{l1-l1 sparse setHBar} gives $\sum_{i=1}^{p}a_{i}^{2}=\bar{h}$.

Let us now denote $\bar{s}_{\ell_{1}\hspace{-1pt}\text{-}\ell_{1}}\hspace{0pt}=\hspace{0pt}\frac{s_{0}\hspace{0pt}+\hspace{0pt}s_{1}}{2}$
and set $\tau\hspace{0pt}=\hspace{0pt}\sqrt{2\log(n/\bar{s}_{\ell_{1}\hspace{-1pt}\text{-}\ell_{1}})}$; thus,
we have
\begin{equation}\label{distanceToSubdifferentialExpectationCSRepresentRelationApplySet} U_{g_{\ell_{1}\text{-}\ell_{1}}}\leq
2\bar{h}\log\frac{n}{\bar{s}_{\ell_{1}\text{-}\ell_{1}}}
+\bar{s}_{\ell_{1}\text{-}\ell_{1}}
+\frac{\bar{s}_{\ell_{1}\text{-}\ell_{1}}(1-\frac{\bar{s}_{\ell_{1}\text{-}\ell_{1}}}
{n})}{\sqrt{2\pi\log(\frac{n}{\bar{s}_{\ell_{1}\text{-}\ell_{1}}})}}.
\end{equation}
Applying \eqref{logInequality} on the third term of the right hand side of \eqref{distanceToSubdifferentialExpectationCSRepresentRelationApplySet}
gives
\begin{equation}\label{distanceToSubdifferentialExpectationCSRepresentRelationApplySetFinal}
U_{g_{\ell_{1}\text{-}\ell_{1}}}\leq 2\bar{h}\log\frac{n}{\bar{s}_{\ell_{1}\text{-}\ell_{1}}}
+\frac{7}{5}\bar{s}_{\ell_{1}\text{-}\ell_{1}}.
\end{equation}

Finally, we obtain the $\ell_1\text{-}\ell_1$ minimization bound~\cite{MotaGLOBALSIP14,MotaARXIV14,MotaICASSP15}
in~\eqref{l1-l1 bound} as
\begin{equation}\label{l1-l1-ramsiBound}
m_{\ell_{1} \text{-}\ell_{1}}\geq 2\bar{h}\log\frac{n}{\bar{s}_{\ell_{1}\text{-}\ell_{1}}} +\frac{7}{5}\bar{s}_{\ell_{1}\text{-}\ell_{1}} + 1,
\end{equation}
where $\bar{s}_{\ell_{1}\text{-}\ell_{1}}=\hspace{0pt}\frac{s_{0}\hspace{0pt}+\hspace{0pt}s_{1}}{2}=s_{0}+\frac{\xi}{2}$,
with $\xi$ defined in \eqref{l1-l1 sparse setXi}.
\end{proof}

\section{}\label{ap:BoundApprox}
\textcolor{black}{\begin{lemma}\label{Proposition delta}
The value of $\delta_{n\text{-}\ell_{1}}$ in the definition of the bound
for weighted $n\text{-}\ell_1$ minimization---given in~\eqref{upperBoundComputeCSApplyAEtaSetMeasurement}---is
negative.
\end{lemma}}
\textcolor{black}{\begin{proof}
Recall the definition of
$\delta_{n\text{-}\ell_{1}}$ given in \eqref{deltaNL1}.
By the definition of $\bar{s}_{n\text{-}\ell_{1}}$ in \eqref{sparseQuantity}, it is
clear that $(\bar{s}_{n\text{-}\ell_{1}}\hspace{-1pt}-\hspace{-1pt}p)>0$; this is because $c_{i}<1$ [see \eqref{c}].
Hence, the sign of $\delta_{n\text{-}\ell_{1}}$ depends on the term $(\kappa_{n\text{-}\ell_{1}}-1)$.
From~\eqref{kappaQuantity}, it is clear that $\kappa_{n\text{-}\ell_{1}}$ depends on $c_{i}$, which is
defined as
 \begin{equation}
\label{cKappaEquality}
  c_{i}=d_{i}\Big(d_{i}+\sum\limits_{j\notin\{l_{i},\dots,l_{i}+d_{i}-1\}}\frac{\epsilon}{|x_{i}-z_{ji}|+\epsilon}\Big)^{-1}.
\end{equation}
As $\epsilon>0$ and very small, we can say that $c_{i}\rightarrow1^-$. As a result, $\kappa_{n\text{-}\ell_{1}}$
is approximately given by
  \begin{equation}\label{kappaEpQuantity}
\kappa_{n\text{-}\ell_{1}} \approx\hspace{-2pt}\frac{4}{\sqrt{2\pi}\tau}\approx\hspace{-2pt}\frac{2}{\sqrt{\pi
\log(n/\bar{s}_{n\text{-}\ell_{1}})}},
\end{equation}
where in the proof of Theorem~\ref{RAMSIBound} we have set that $\tau\hspace{-2pt}=\hspace{-2pt}\sqrt{2 \log(n/\bar{s}_{n\text{-}\ell_{1}})}$.
We observe that $\kappa_{n\text{-}\ell_{1}} < 1$ if $\frac{\bar{s}_{n\text{-}\ell_{1}}} {n}<0.28$, where $\bar{s}_{n\text{-}\ell_{1}}\hspace{-2pt}\approx
\hspace{-2pt}p\leq\min\{s_{j}\}$ [see \eqref{sparseQuantity} and Assumption \ref{asIndex}].
In Lemma \ref{propositionBound}, we prove that when $\frac{p}{n}<0.23$ then $\widehat{m}_{n\text{-}\ell_{1}}$ in \eqref{boundCSApproLoosest}
is less than the source dimension, $n$. Otherwise, the required number of measurements $\widehat{m}_{n\text{-}\ell_{1}}$
is higher than $n$, signifying the failure of the algorithm\footnote{Indeed, when the source signal is not sparse enough and the side information signals are not highly correlated with the source, the recovery algorithm will have a low performance.}.
Hence, we get that $(\kappa_{n\text{-}\ell_{1}}-1)<0$, thereby proving that $\delta_{n\text{-}\ell_{1}}<0 $.\end{proof}}

\begin{lemma}\label{propositionBound}
Given a sparse signal $\bx\hspace{-2pt}\in \hspace{-2pt}\mathbb{R}^{n}$ with $\|\bx\|_0=s_{0}$ and side information signals $\bz_{j}\in\mathbb{R}^{n}$ with $\|\bx-\bz_{j}\|_0=s_{j}$, $\forall j\in\{1,\dots,J\}$, the number of measurements required for weighted $n\text{-}\ell_{1}$ minimization to recover $\bx$---given in \eqref{boundCSApproLoosest}---is less than the source dimension $n$ when $\frac{p}{n}<0.23$.\end{lemma}
\begin{proof}
\textcolor{black}{Let us assume that the number of measurements in the bound of \eqref{boundCSApproLoosest}
satisfies the condition $\widehat{m}_{n\text{-}\ell_{1}}
< n$, namely, we have
\begin{equation}\label{l1 boundCompare}
2p\log\frac{n}{p} + \frac{7}{5}p < n.
\end{equation}
By substituting $\gamma=\frac{n}{p}$, $\gamma\in\mathbb{R}$, in~\eqref{l1 boundCompare} we get
\begin{equation}\label{l1 boundCompareX}
2\log \gamma + \frac{7}{5} < \gamma
\end{equation}
By setting $f(\gamma)=\gamma-2\log \gamma - \frac{7}{5}$ and using calculus we can show that $f(\gamma)>0$ when $\gamma>4.33$, or else, $\frac{p}{n}<0.23$.}
\end{proof}

\bibliographystyle{IEEEtran}
\bibliography{./IEEEfull,./IEEEabrv,./bibliography}

\begin{thebibliography}{10}
\providecommand{\url}[1]{#1}
\csname url@samestyle\endcsname
\providecommand{\newblock}{\relax}
\providecommand{\bibinfo}[2]{#2}
\providecommand{\BIBentrySTDinterwordspacing}{\spaceskip=0pt\relax}
\providecommand{\BIBentryALTinterwordstretchfactor}{4}
\providecommand{\BIBentryALTinterwordspacing}{\spaceskip=\fontdimen2\font plus
\BIBentryALTinterwordstretchfactor\fontdimen3\font minus
  \fontdimen4\font\relax}
\providecommand{\BIBforeignlanguage}[2]{{%
\expandafter\ifx\csname l@#1\endcsname\relax
\typeout{** WARNING: IEEEtran.bst: No hyphenation pattern has been}%
\typeout{** loaded for the language `#1'. Using the pattern for}%
\typeout{** the default language instead.}%
\else
\language=\csname l@#1\endcsname
\fi
#2}}
\providecommand{\BIBdecl}{\relax}
\BIBdecl

\bibitem{DonohoTIT06}
D.~Donoho, ``Compressed sensing,'' \emph{{IEEE} Trans. Inf. Theory}, vol.~52,
  no.~4, pp. 1289--1306, Apr. 2006.

\bibitem{candes2006robust}
E.~J. Cand{\`e}s, J.~Romberg, and T.~Tao, ``Robust uncertainty principles:
  Exact signal reconstruction from highly incomplete frequency information,''
  \emph{{IEEE} Trans. Inf. Theory}, vol.~52, no.~2, pp. 489--509, 2006.

\bibitem{DonohoCOM06}
D.~Donoho, ``For most large underdetermined systems of linear equations the
  minimal $\ell_{1}$-norm solution is also the sparsest solution,''
  \emph{Communications on Pure and Applied Math}, vol.~59, no.~6, pp. 797--829,
  2006.

\bibitem{CandesTIT06}
E.~Cand\`{e}s and T.~Tao, ``Near-optimal signal recovery from random
  projections: Universal encoding strategies?'' \emph{{IEEE} Trans. Inf.
  Theory}, vol.~52, no.~12, pp. 5406--5425, Apr. 2006.

\bibitem{Beck09}
A.~Beck and M.~Teboulle, ``A fast iterative shrinkage-thresholding algorithm
  for linear inverse problems,'' \emph{SIAM Journal on Imaging Sciences}, vol.
  2(1), pp. 183--202, 2009.

\bibitem{Candes08}
E.~Cand\`{e}s, M.~B. Wakin, , and S.~P. Boyd, ``Enhancing sparsity by
  reweighted $\ell_{1}$ minimization,'' \emph{J. Fourier Anal. Appl.}, vol.~14,
  no. 5-6, pp. 877--905, 2008.

\bibitem{Asif13}
M.~S. Asif and J.~Romberg, ``Fast and accurate algorithms for re-weighted
  $\ell_{1}$-norm minimization,'' \emph{{IEEE} Trans. Signal Process.},
  vol.~61, no.~23, pp. 5905--5916, Dec. 2013.

\bibitem{NVaswani09}
N.~Vaswani and W.~Lu, ``Modified-cs: Modifying compressive sensing for problems
  with partially known support,'' in \emph{{IEEE} Int. Symposium on Information
  Theory}, Seoul, Korea, Jul. 2009.

\bibitem{NVaswani10}
N.~Vaswani, ``Stability (over time) of modified-cs for recursive causal sparse
  reconstruction,'' in \emph{Allerton Conf. Communications, Control, and
  Computing}, Monticello, Illinois, USA, Sep. 2010.

\bibitem{NVaswaniTSP10}
N.~Vaswani and W.~Lu, ``Modified-cs: Modifying compressive sensing for problems
  with partially known support,'' \emph{{IEEE} Trans. Signal Process.},
  vol.~58, no.~9, pp. 4595--4607, Sep. 2010.

\bibitem{VaswaniZ16}
N.~Vaswani and J.~Zhan, ``Recursive recovery of sparse signal sequences from
  compressive measurements: {A} review,'' \emph{{IEEE} Trans. Signal Process.},
  vol.~64, no.~13, pp. 3523--3549, 2016.

\bibitem{MPFriedlander12}
M.~P. Friedlander, H.~Mansour, R.~Saab, and O.~Yilmaz, ``Recovering
  compressively sampled signals using partial support information,''
  \emph{{IEEE} Trans. Inf. Theory}, vol.~58, no.~2, p. 1122–1134, Feb. 2012.

\bibitem{JZhan14}
J.~Zhan and N.~Vaswani, ``Robust pca with partial subspace knowledge,'' in
  \emph{IEEE Intern. Symposium on Information Theory (ISIT)}, Hawaii, USA, Jun.
  2010.

\bibitem{JZhan15}
------, ``Time invariant error bounds for modified-cs-based sparse signal
  sequence recovery,'' \emph{{IEEE} Trans. Inf. Theory}, vol.~61, no.~3, pp.
  1389--1409, Mar. 2015.

\bibitem{AminTIT15}
M.~A. Khajehnejad, W.~Xu, A.~S. Avestimehr, and B.~Hassibi, ``Improving the
  thresholds of sparse recovery: An analysis of a two-step reweighted basis
  pursuit algorithm,'' \emph{{IEEE} Trans. Inf. Theory}, vol.~61, no.~9, pp.
  5116--5128, Sep. 2015.

\bibitem{Weizman15}
L.~Weizman, Y.~C. Eldar, and D.~B. Bashat, ``Compressed sensing for longitunal
  \textsc{MRI}: An adaptive-weighted approach,'' \emph{Medical Physics},
  vol.~42, no.~9, pp. 5195--5207, 2015.

\bibitem{Venkat12}
V.~Chandrasekaran, B.~Recht, P.~A. Parrilo, and A.~S. Willsky, ``The convex
  geometry of linear inverse problems,'' \emph{Foundations of Computational
  Mathematics}, vol.~12, no.~6, pp. 805--849, 2012.

\bibitem{Tropp14}
D.~Amelunxen, M.~Lotz, M.~B. McCoy, and J.~A. Tropp, ``Living on the edge:
  phase transitions in convex programs with random data,'' \emph{Information
  and Inference}, vol.~3, no.~3, pp. 224--294, 2014.

\bibitem{BaronARXIV09}
D.~Baron, M.~F. Duarte, M.~B. Wakin, S.~Sarvotham, and R.~G. Baraniuk,
  ``Distributed compressive sensing,'' ArXiv e-prints, Jan. 2009.

\bibitem{DuarteTIT13}
M.~Duarte, M.~Wakin, D.~Baron, S.~Sarvotham, and R.~Baraniuk, ``Measurement
  bounds for sparse signal ensembles via graphical models,'' \emph{{IEEE}
  Trans. Inf. Theory}, vol.~59, no.~7, pp. 4280--4289, Jul. 2013.

\bibitem{NVaswaniTSP12}
W.~Lu and N.~Vaswani, ``Regularized modified {BPDN} for noisy sparse
  reconstruction with partial erroneous support and signal value knowledge,''
  \emph{{IEEE} Trans. Signal Process.}, vol.~60, no.~1, pp. 182--196, 2012.

\bibitem{ZhangR11}
Z.~Zhang and B.~D. Rao, ``Sparse signal recovery with temporally correlated
  source vectors using sparse bayesian learning,'' \emph{J. Sel. Topics Signal
  Processing}, vol.~5, no.~5, pp. 912--926, 2011.

\bibitem{KhajehnejadXAH09}
M.~A. Khajehnejad, W.~Xu, A.~S. Avestimehr, and B.~Hassibi, ``Weighted
  $\ell_{1}$ minimization for sparse recovery with prior information,'' in
  \emph{{IEEE} Int. Symposium on Information Theory}, Seoul, Korea, Jul. 2009.

\bibitem{KhajehnejadXAH11}
------, ``Analyzing weighted $\ell_{1}$ minimization for sparse recovery with
  nonuniform sparse models,'' \emph{{IEEE} Trans. Signal Process.}, vol.~59,
  no.~5, pp. 1985--2001, 2011.

\bibitem{MotaGLOBALSIP14}
J.~F. Mota, N.~Deligiannis, and M.~R. Rodrigues, ``Compressed sensing with side
  information: Geometrical interpretation and performance bounds,'' in
  \emph{IEEE Global Conf. on Signal and Information Processing}, Austin, Texas,
  USA, Dec. 2014.

\bibitem{MotaARXIV14}
------, ``Compressed sensing with prior information: Optimal strategies,
  geometry, and bounds,'' ArXiv e-prints, Aug. 2014.

\bibitem{MansourS15}
H.~Mansour and R.~Saab, ``Weighted one-norm minimization with inaccurate
  support estimates: Sharp analysis via the null-space property,'' in
  \emph{IEEE Int. Conf. on Acoustics, Speech and Signal Processing}, Brisbane,
  Australia, Apr. 2015.

\bibitem{MotaICASSP15}
J.~F. Mota, N.~Deligiannis, A.~Sankaranarayanan, V.~Cevher, and M.~R.
  Rodrigues, ``Dynamic sparse state estimation using $\ell_{1}$-$\ell_{1}$
  minimization: Adaptive-rate measurement bounds, algorithms and
  applications,'' in \emph{IEEE Int. Conf. on Acoustics, Speech and Signal
  Processing}, Brisbane, Australia, Apr. 2015.

\bibitem{MotaARXIV15}
J.~F. Mota, N.~Deligiannis, A.~C. Sankaranarayanan, V.~Cevher, and M.~R.
  Rodrigues, ``Adaptive-rate sparse signal reconstruction with application in
  compressive background subtraction,'' \emph{{IEEE} Trans. Signal Process.},
  in press, 2016.

\bibitem{Becquaert2016Radar}
M.~Becquaert, E.~Cristofani, G.~Pandey, M.~Vandewal, J.~Stiens, and
  N.~Deligiannis, ``Compressed sensing mm-wave {SAR} for non-destructive
  testing applications using side information,'' in \emph{IEEE Radar
  Conference}, Philadelphia, PA, May 2016.

\bibitem{Scarlett}
J.~Scarlett, J.~Evans, and S.~Dey, ``Compressed sensing with prior information:
  Information-theoretic limits and practical decoders,'' \emph{{IEEE} Trans.
  Signal Process.}, vol.~61, no.~2, pp. 427--439, Jan. 2013.

\bibitem{AsifR14}
M.~S. Asif and J.~K. Romberg, ``Sparse recovery of streaming signals using
  $\ell_{1}$-homotopy,'' \emph{{IEEE} Trans. Signal Process.}, vol.~62, no.~16,
  pp. 4209--4223, 2014.

\bibitem{QiuV11}
C.~Qiu and N.~Vaswani, ``Support-predicted modified-cs for recursive robust
  principal components' pursuit,'' in \emph{{IEEE} Int. Symposium on
  Information Theory}, St. Petersburg, Russia, Jul. 2011.

\bibitem{GuoQV14}
H.~Guo, C.~Qiu, and N.~Vaswani, ``An online algorithm for separating sparse and
  low-dimensional signal sequences from their sum,'' \emph{{IEEE} Trans. Signal
  Process.}, vol.~62, no.~16, pp. 4284--4297, 2014.

\bibitem{QiuVLH14}
C.~Qiu, N.~Vaswani, B.~Lois, and L.~Hogben, ``Recursive robust {PCA} or
  recursive sparse recovery in large but structured noise,'' \emph{{IEEE}
  Trans. Inf. Theory}, vol.~60, no.~8, pp. 5007--5039, 2014.

\bibitem{Taj2011}
M.~Taj and A.~Cavallaro, ``Distributed and decentralized multicamera
  tracking,'' \emph{{IEEE} Signal Process. Mag.}, vol.~28, no.~3, pp. 46--58,
  2011.

\bibitem{ChellappaIEEE08}
A.~Sankaranarayanan, A.~Veeraraghavan, and R.~Chellappa, ``Object detection,
  tracking and recognition for multiple smart cameras,'' \emph{Proc. of IEEE},
  vol.~96, no.~10, pp. 1606--1624, 2008.

\bibitem{YangIEEE10}
A.~Y. Yang, M.~Gastpar, R.~Bajcsy, and S.~Sastry, ``Distributed sensor
  perception via sparse representation,'' \emph{Proc. of IEEE}, vol.~98, no.~6,
  pp. 1077--1088, 2010.

\bibitem{LuongDCC16}
H.~V. Luong, J.~Seiler, A.~Kaup, and S.~Forchhammer, ``A reconstruction
  algorithm with multiple side information for distributed compression of
  sparse sources,'' in \emph{Data Compression Conference}, Snowbird, Utah, USA,
  Apr. 2016.

\bibitem{WarnellTIP15}
G.~Warnell, S.~Bhattacharya, R.~Chellappa, and T.~Basar, ``Adaptive-rate
  compressive sensing using side information,'' \emph{{IEEE} Trans. Image
  Process.}, vol.~24, no.~11, pp. 3846--3857, 2015.

\bibitem{JHiriat}
J.-B. Hiriart-Urruty and C.~Lemar\'{e}chal, \emph{Fundamentals of Convex
  Analysis}.\hskip 1em plus 0.5em minus 0.4em\relax Springer, 2004.

\bibitem{Nene96}
S.~A. Nene, S.~K. Nayar, and H.~Murase, ``Columbia object image library
  (coil-100),'' Technical Report CUCS-006-96, Feb. 1996.

\bibitem{DLowe99}
D.~G. Lowe, ``Object recognition from local scale-invariant features,'' in
  \emph{IEEE Int. Conf. on Computer Vision}, Kerkyra, Greece, Sep. 1999.

\bibitem{NisterCVPR06}
D.~Nist\'{e}r and H.~Stew\'{e}nius, ``Scalable recognition with a vocabulary
  tree,'' in \emph{IEEE Int. Conf. on Computer Vision and Pattern Recognition},
  New York, USA, Jun. 2006.

\bibitem{bachmann2012fourier}
G.~Bachmann, L.~Narici, and E.~Beckenstein, \emph{Fourier and wavelet
  analysis}.\hskip 1em plus 0.5em minus 0.4em\relax Springer Science \&
  Business Media, 2012.

\end{thebibliography}

\end{document}